\documentclass[11pt]{article}
\usepackage[margin=1in]{geometry}
\usepackage{amsmath,amssymb,amsthm,mathtools}
\usepackage{algorithm}
\usepackage{algpseudocode}
\usepackage{natbib}
\usepackage{enumitem}
\usepackage{graphicx}
\usepackage{csvsimple}
\usepackage{longtable}
\usepackage{booktabs}
\usepackage{caption}
\usepackage{authblk}

\newtheorem{assumption}{Assumption}

\newtheorem{theorem}{Theorem}
\newtheorem{proposition}{Proposition}
\newtheorem{lemma}{Lemma}
\newtheorem{remark}{Remark}

\newcommand{\R}{\mathbb{R}}

\newcommand{\PP}{\mathbb{P}}
\newcommand{\EE}{\mathbb{E}}
\newcommand{\ind}{\mathbf{1}}
\newcommand{\calR}{\mathcal{R}}
\newcommand{\calT}{\mathcal{T}}
\newcommand{\calC}{\mathcal{C}}
\newcommand{\calK}{\mathcal{K}}
\newcommand{\PiSplit}{\Pi}
\newcommand{\argmin}{\operatorname*{arg\,min}}
\newcommand{\argmax}{\operatorname*{arg\,max}}

\usepackage{hyperref}
\hypersetup{
    colorlinks=true,
    linkcolor=blue,
    citecolor=blue,
    urlcolor=blue
}

\numberwithin{equation}{section}

\newcommand{\Cov}{\operatorname{Cov}}

\title{Distribution-Free Changepoint Inference for Rainfall Patterns}

\author[1]{Aaditya Jain\thanks{aadi.jain162@gmail.com}}
\author[2]{Abhishek Bhattacharjee\thanks{abhishek@theabstractmath.com}}

\affil[1,2]{Abstract Math Institute}
\affil[1]{\texttt{aadi.jain162@gmail.com}}
\affil[2]{\texttt{abhishek@theabstractmath.com}}

\begin{document}

\maketitle

\section{Introduction}
\label{sec:introduction}

Long records of rainfall, humidity, and related hydroclimatic variables are increasingly available from monitoring networks that collect daily measurements over many years. These data contain information not only about annual totals or seasonal averages, but also about the temporal organization of rainfall within a year: the timing of wet and dry spells, the strength of periodic components, and the distribution of variation across intra-annual frequencies. Detecting whether, and when, this temporal pattern changes is a central problem in climate-impact assessment, water-resource planning, agricultural risk management, and environmental monitoring.

A common data structure consists of observations from multiple monitoring locations over a long sequence of years. For each location and each year, the observed record is a vector of daily measurements, such as rainfall or humidity, of length approximately \(365\). Let \(m\) denote the number of monitoring locations and let \(T\) denote the number of years. The data may then be represented as
\[
\bigl\{X_{i,t} = (X_{i,t,1},\ldots,X_{i,t,D})^\top:
    i=1,\ldots,m,\ t=1,\ldots,T\bigr\},
\]
where \(i\) indexes location, \(t\) indexes year, and \(D=365\) for non-leap years. The inferential goal is to identify an unknown year \(t_0\) at which the temporal rainfall pattern changes across the region. Unlike a change in a scalar annual summary, a structural change in the intra-annual pattern may affect the frequency-domain behavior of the daily sequence within each year. This motivates representing each yearly curve through a spectral feature that captures the strength of temporal oscillation at a fixed frequency.

Classical changepoint methods often rely on parametric likelihoods, Gaussian approximations, large-sample calibration, or model-specific assumptions on the distribution of the observations. Such assumptions can be difficult to justify for rainfall data, which are typically skewed, intermittent, heavy-tailed, seasonally heterogeneous, and spatially variable. Moreover, many procedures return only a point estimate of the changepoint, although in environmental applications a confidence set for the change year is often more useful: it quantifies the temporal uncertainty in the estimated onset of a structural shift and allows practitioners to distinguish sharp transitions from weak or gradual evidence.

This paper develops a distribution-free framework for detecting a common structural change in yearly rainfall patterns across multiple independently monitored locations. For each location-year pair, we compute a fixed-frequency nonparametric spectral density estimate from the daily record. These scalar spectral features form an \(T\times m\) data matrix, with rows corresponding to years and columns corresponding to locations. The changepoint problem is then formulated as a distribution-free inference problem for a synchronized distributional change in the columns of this matrix. The method exploits exchangeability of the yearly spectral features before and after the true changepoint, together with independence across monitoring locations, to construct finite-sample valid confidence sets for the unknown change year.

\subsection{Our Contributions}
\label{subsec:our-contributions}

In this paper, our contributions are as follows:
\begin{itemize}
    \item We introduce a frequency-domain framework for temporal rainfall pattern change detection. Instead of reducing each year to an annual total, annual mean, or prespecified seasonal index, we represent each yearly daily record through a nonparametric estimate of the spectral density at a fixed frequency \(\theta\in[0,2\pi]\). This representation targets changes in the intra-annual organization of rainfall and provides a principled scalar feature for each location-year pair.

    \item We develop a distribution-free multi-location changepoint procedure for detecting a synchronized structural change in rainfall patterns across a monitoring network. The yearly spectral features are arranged as a \(T\times m\) matrix, with years as rows and monitoring locations as columns, and the changepoint problem is formulated as inference on a common distributional change across the columns of this matrix.

    \item We construct confidence sets for the unknown change year rather than reporting only a point estimate. The proposed method evaluates candidate changepoints through conformal split-permutation \(p\)-values and returns the set of years that cannot be rejected at a prescribed significance level. This output is directly interpretable for environmental applications, where uncertainty in the timing of a structural transition is scientifically relevant.

    \item We establish finite-sample distribution-free validity under cross-location independence and segment-wise exchangeability of the yearly spectral features before and after the true changepoint. The coverage guarantee does not rely on Gaussianity, parametric rainfall models, large-sample approximations, or consistent estimation of nuisance distributions.

    \item We provide a rigorous integration of nonparametric spectral estimation with distribution-free changepoint inference. The spectral step captures changes in temporal rainfall patterns within each year, while the conformal changepoint step supplies finite-sample valid uncertainty quantification for the common change year.

    \item We study the empirical behavior of the method through simulation experiments that assess coverage, localization accuracy, and sensitivity to structural changes in the intra-annual rainfall pattern.
\end{itemize}

\paragraph{Structure of the paper.}
The remainder of the paper is organized as follows. Section~\ref{sec:framework_method} introduces the statistical framework, defines the spectral rainfall features, presents the proposed distribution-free changepoint algorithm, and establishes its finite-sample theoretical guarantees. Section~\ref{sec:simulation-studies} reports simulation studies designed to evaluate the empirical coverage, localization behavior, and sensitivity of the method under controlled structural changes in temporal rainfall patterns. Section~\ref{sec:spatial} extends the validity guarantee to spatially dependent stations. Section~\ref{sec:conclusion} concludes with a discussion of the scope of the method, its practical implications, and directions for future work.

\section{Existing Works}
\label{sec:existing_works}

The problem of detecting structural breaks in environmental time series intersects several distinct areas of statistical literature, ranging from classical changepoint detection to modern distribution-free inference. 

\paragraph{Classical and Parametric Changepoint Detection.}
Traditional approaches to changepoint detection typically rely on cumulative sum (CUSUM) statistics, likelihood ratio tests, or penalized cost functions \citep{page1954continuous, killick2012optimal}. While these methods are computationally efficient and well-understood asymptotically, they frequently require parametric assumptions—such as Gaussianity or specific exponential family distributions—or rely on large-sample approximations to calibrate critical values. In the context of hydroclimatic data, these assumptions are often violated. Rainfall records are notoriously intermittent, highly skewed, and heavy-tailed, making asymptotic calibration highly unreliable in finite samples and leading to inflated false positive rates or invalid confidence intervals \citep{reeves2007review}.

\paragraph{Functional and Spectral Time Series Analysis.}
To capture changes in intra-annual patterns rather than just annual scalar aggregates, recent literature has shifted toward functional data analysis (FDA) and spectral methods. Treating daily yearly records as functional curves allows researchers to detect shape or phase changes over time \citep{aston2012evaluating, aue2009estimation}. In the frequency domain, spectral density estimation provides a powerful tool for isolating changes in the periodic structure of time series and random fields \citep{deb2017asymptotic}. However, while spectral and functional methods excel at feature representation, inferring the exact timing of a changepoint and constructing valid confidence sets for these high-dimensional representations typically requires complex bootstrap procedures or reliance on asymptotic distributions that may struggle in finite networks.

\paragraph{Distribution-Free and Conformal Inference.}
To circumvent the limitations of parametric and asymptotic calibration, there has been a recent surge in distribution-free inference techniques. Leveraging exchangeability, conformal inference and permutation-based methods provide exact, finite-sample guarantees for predictive intervals and hypothesis testing without specifying an underlying data-generating mechanism \citep{vovk2005algorithmic, lei2018distribution}. Recently, these ideas have been extended to changepoint and root-cause analysis via split-permutation techniques \citep{hore2026distributionfree}. By integrating robust spectral feature extraction with conformal split-permutation \(p\)-values, our proposed framework bridges the gap between complex temporal pattern recognition and rigorous, finite-sample uncertainty quantification for environmental monitoring networks.

\section{Framework and Our Method}\label{sec:framework_method}

We observe annual rainfall records from $m$ spatially separated stations over $n$ calendar years.  The word station will be used for the physical location or the server attached to that location.  For station $j\in[m]:=\{1,\ldots,m\}$ and year $i\in[n]$, let
\[
        Y_{i,j}=(Y_{i,j}(1),\ldots,Y_{i,j}(q))\in\R^q
\]
denote the within-year record, where $q=365$ after applying a fixed calendar convention for leap years and missing days.  The coordinate $r\in[q]$ indexes day $r$ of the year.  The method below is stated for a scalar daily measurement, such as rainfall or humidity.  If several meteorological variables are available, $Y_{i,j}(r)$ may be replaced by any pre-specified scalar transformation of the multivariate record, or the construction may be repeated over several coordinates and frequencies before combining the resulting scores.  All such transformations must be fixed independently of the split-permutation step described below.

The target is a temporal structural change in the annual rainfall pattern.  We allow the station-specific changepoint vector to be
\[
        \xi=(\xi_1,\ldots,\xi_m)\in \calR\subseteq \{1,\ldots,n\}^m,
\]
where $\xi_j=n$ is interpreted as no change at station $j$ over the observation window.  Thus $\xi_j$ is the last pre-change year at station $j$, and the first post-change year is $\xi_j+1$ whenever $\xi_j<n$; if one prefers to report the first changed calendar year, all reported changepoints should be shifted by one.  The synchronized regional model, which is the primary case in this paper, imposes
\[
        \calR_{\mathrm{syn}}=\{(\tau,\ldots,\tau):\tau\in\calT\},\qquad
        \calT=\{\ell,\ell+1,\ldots,n-\ell\},
\]
for a fixed trimming parameter $\ell\geq 1$.  In this case the unknown regional change year is $\tau_0$, with $\xi=(\tau_0,\ldots,\tau_0)$.  More general sets $\calR$ can encode delayed propagation, excluded boundary years, or station-specific prior information.  For example, a bounded-lag model can be written as $\calR=\{t\in\calT^m:\max_j t_j-\min_j t_j\leq L\}$ for a known lag $L$.

The distributional assumption is deliberately weak.  For each station $j$, conditional on the changepoint $\xi_j$, the pre-change block $(Y_{1,j},\ldots,Y_{\xi_j,j})$ is exchangeable, the post-change block $(Y_{\xi_j+1,j},\ldots,Y_{n,j})$ is exchangeable, and these two blocks are independent.  The $m$ station-level arrays are mutually independent.  No parametric form is imposed on the distribution of the daily vectors, and the pre- and post-change laws may differ across stations.  The synchronized model adds only the restriction that the $m$ changepoints are equal.  This is the assumption under which the finite-sample confidence sets below are valid.  Dependence across days within a year is unrestricted for validity, because each full annual vector is treated as one observation before being reduced to a spectral summary.  Conditions on within-year dependence are needed only if the spectral summary is to be interpreted as a consistent estimator of an underlying spectral density.

We now define the spectral feature used for each station-year pair.  Fix a frequency $\theta\in[0,2\pi]$, a symmetric kernel $K:\R\to\R$ satisfying $K(0)=1$, and a bandwidth $B_q>0$.  Let
\[
        \bar Y_{i,j}=q^{-1}\sum_{r=1}^q Y_{i,j}(r),\qquad
        X_{i,j}(r)=Y_{i,j}(r)-\bar Y_{i,j}.
\]
The centering is deterministic within each year and can be omitted or replaced by another pre-specified centering convention if the scientific target requires the annual mean to be retained.  The fixed-frequency nonparametric spectral feature is
\begin{equation}\label{eq:spectral_feature}
        Z_{i,j}(\theta)
        =\widehat f_{i,j,q}(\theta)
        :=\frac{1}{q}\sum_{r=1}^q\sum_{s=1}^q
        X_{i,j}(r)X_{i,j}(s)
        K\!\left(\frac{r-s}{B_q}\right)\cos\{(r-s)\theta\}.
\end{equation}
This is the one-dimensional regular-grid version of the kernel spectral density estimator for random fields.  The cosine form is equivalent to the real part of the complex exponential form and is real-valued for real rainfall records.  For a set of frequencies $\Theta=\{\theta_1,\ldots,\theta_L\}$, one may replace $Z_{i,j}(\theta)$ by the vector $(Z_{i,j}(\theta_1),\ldots,Z_{i,j}(\theta_L))$ and use the same conformal construction with a multivariate score.  To keep notation uncluttered, the rest of the section is written for a single fixed $\theta$ and we set $Z_{i,j}=Z_{i,j}(\theta)$.

The reduction from $Y_{i,j}$ to $Z_{i,j}$ preserves the exchangeability needed for distribution-free inference.  Indeed, for each station $j$, the map $Y_{i,j}\mapsto Z_{i,j}$ is applied identically to every year and is deterministic.  Hence, if $(Y_{1,j},\ldots,Y_{\xi_j,j})$ is exchangeable, then $(Z_{1,j},\ldots,Z_{\xi_j,j})$ is exchangeable; the same holds after the change.  The analysis may therefore be conducted on the $n\times m$ matrix
\[
        Z_\theta=(Z_{i,j})_{i\in[n],\,j\in[m]},
\]
whose columns are the station-level annual spectral sequences.

For a candidate changepoint vector $t=(t_1,\ldots,t_m)\in\calR$, define the split-permutation group
\begin{equation}\label{eq:split_group}
        \PiSplit_t=\bigg\{\pi=(\pi_1,\ldots,\pi_m):
        \begin{array}{l}
        \pi_j \text{ is a permutation of } [n],\;\pi_j(i)\leq t_j \text{ for } i\leq t_j,\\
        \pi_j(i)>t_j \text{ for } i>t_j,
        \text{ for every } j\in[m]
        \end{array}\bigg\}.
\end{equation}
For $\pi\in\PiSplit_t$, $\pi(Z_\theta)$ denotes the matrix obtained by replacing column $j$ by $(Z_{\pi_j(1),j},\ldots,Z_{\pi_j(n),j})$.  If $t=\xi$, then the distribution of $Z_\theta$ is invariant under every permutation in $\PiSplit_t$.  This invariance is the sole source of finite-sample validity.

The method requires a changepoint-plausibility score $S_\theta(Z,t)$, where larger values indicate that the candidate $t$ is more compatible with the ordered spectral-feature matrix $Z$.  A useful choice is obtained by first constructing a changepoint estimate from the ordered columns and then scoring candidates according to their distance from that estimate.  For $s\in\calT$ and station $j$, write
\[
        \bar Z_{1:s,j}=s^{-1}\sum_{i=1}^s Z_{i,j},\qquad
        \bar Z_{s+1:n,j}=(n-s)^{-1}\sum_{i=s+1}^n Z_{i,j},
\]
and define the absolute CUSUM contrast
\begin{equation}\label{eq:station_cusum}
        C_j(s;Z)=\left\{\frac{s(n-s)}{n}\right\}^{1/2}
        \left|\bar Z_{1:s,j}-\bar Z_{s+1:n,j}\right|.
\end{equation}
For the unrestricted or bounded-lag model, define
\[
        \widehat\xi_j(Z)=\min\argmax_{s\in\calT} C_j(s;Z),\qquad
        \widehat\xi(Z)=(\widehat\xi_1(Z),\ldots,\widehat\xi_m(Z)),
\]
where the minimum imposes deterministic tie-breaking.  Given positive station weights $w_1,\ldots,w_m$ with $\sum_jw_j=1$, set
\begin{equation}\label{eq:general_score}
        S_\theta(Z,t)=-\sum_{j=1}^m w_j\,|\widehat\xi_j(Z)-t_j|.
\end{equation}
In the synchronized regional model, it is often more stable to use a pooled contrast
\begin{equation}\label{eq:pooled_cusum}
        C_{\mathrm{pool}}(s;Z)=\sum_{j=1}^m w_j C_j(s;Z),\qquad
        \widehat\tau(Z)=\min\argmax_{s\in\calT} C_{\mathrm{pool}}(s;Z),
\end{equation}
and then score $t=(\tau,\ldots,\tau)$ by
\begin{equation}\label{eq:syn_score}
        S_\theta\{Z,(\tau,\ldots,\tau)\}=-|\widehat\tau(Z)-\tau|.
\end{equation}
The score in \eqref{eq:general_score} or \eqref{eq:syn_score} is not a model assumption.  It affects the sharpness of the resulting confidence set, but not its validity.  It is important that the score uses the chronological order of the annual features.  Scores that are invariant under all split permutations in \eqref{eq:split_group} lead to uninformative conformal ranks; for example, a statistic depending only on the two unordered multisets induced by the candidate split would produce trivial $p$-values.

For each $t\in\calR$, define the exact conformal $p$-value
\begin{equation}\label{eq:exact_pvalue}
        p_t=\frac{1}{|\PiSplit_t|}\sum_{\pi\in\PiSplit_t}
        \ind\{S_\theta(\pi(Z_\theta),t)\leq S_\theta(Z_\theta,t)\}.
\end{equation}
The inequality direction follows the convention that larger scores mean greater plausibility of $t$.  Thus, $p_t$ is small when the observed ordered spectral matrix makes $t$ unusually implausible relative to the split-permutation reference distribution.  Exact enumeration of $\PiSplit_t$ is usually infeasible.  If $M$ independent permutations $\pi^{(1)},\ldots,\pi^{(M)}$ are drawn uniformly from $\PiSplit_t$, we use the Monte Carlo version
\begin{equation}\label{eq:mc_pvalue}
        \widetilde p_t=\frac{1+\sum_{b=1}^M
        \ind\{S_\theta(\pi^{(b)}(Z_\theta),t)\leq S_\theta(Z_\theta,t)\}}{M+1}.
\end{equation}
The add-one correction is used to preserve finite-sample validity under randomized sampling of permutations.  Below $p_t$ denotes either the exact value in \eqref{eq:exact_pvalue} or the Monte Carlo value in \eqref{eq:mc_pvalue}, with the distinction made explicit when needed.

The confidence set for the changepoint configuration is obtained by inverting the candidate-wise conformal tests:
\begin{equation}\label{eq:configuration_confidence_set}
        \calC_{1-\alpha}^{\calR}(\theta)=\{t\in\calR:p_t>\alpha\}.
\end{equation}
For synchronized rainfall change detection, this reduces to the year-level confidence set
\begin{equation}\label{eq:year_confidence_set}
        \calC_{1-\alpha}^{\mathrm{year}}(\theta)
        =\{\tau\in\calT:(\tau,\ldots,\tau)\in \calC_{1-\alpha}^{\calR_{\mathrm{syn}}}(\theta)\}.
\end{equation}
A point estimate, if desired for reporting, can be taken as $\widehat\tau(Z_\theta)$ or as the element of $\calC_{1-\alpha}^{\mathrm{year}}(\theta)$ maximizing the observed score.  The inferential object of the method is nevertheless the confidence set \eqref{eq:year_confidence_set}, because it is the object with finite-sample distribution-free coverage.

\begin{algorithm}[t]
\caption{Spectral conformal rainfall changepoint inference}\label{alg:spectral_croc}
\begin{algorithmic}[1]
\State \textbf{Input:} annual records $\{Y_{i,j}\in\R^q:i\in[n],j\in[m]\}$, frequency $\theta$, kernel $K$, bandwidth $B_q$, candidate set $\calR$, level $1-\alpha$, trimming set $\calT$, number of Monte Carlo permutations $M$.
\For{$j=1,\ldots,m$ and $i=1,\ldots,n$}
    \State Compute $Z_{i,j}=\widehat f_{i,j,q}(\theta)$ from \eqref{eq:spectral_feature}.
\EndFor
\State Form $Z_\theta=(Z_{i,j})_{i\in[n],j\in[m]}$.
\State Compute the ordered-data changepoint estimate $\widehat\xi(Z_\theta)$ using \eqref{eq:station_cusum}, or compute $\widehat\tau(Z_\theta)$ using \eqref{eq:pooled_cusum} under $\calR=\calR_{\mathrm{syn}}$.
\For{each $t\in\calR$}
    \State Construct the split-permutation group $\PiSplit_t$ in \eqref{eq:split_group}.
    \State Evaluate the observed score $S_\theta(Z_\theta,t)$ using \eqref{eq:general_score} or \eqref{eq:syn_score}.
    \If{exact enumeration is feasible}
        \State Compute $p_t$ using \eqref{eq:exact_pvalue}.
    \Else
        \State Draw $\pi^{(1)},\ldots,\pi^{(M)}\overset{\mathrm{iid}}{\sim}\mathrm{Unif}(\PiSplit_t)$ and compute $\widetilde p_t$ using \eqref{eq:mc_pvalue}.
        \State Set $p_t\leftarrow\widetilde p_t$.
    \EndIf
\EndFor
\State Return $\calC_{1-\alpha}^{\calR}(\theta)=\{t\in\calR:p_t>\alpha\}$.
\If{$\calR=\calR_{\mathrm{syn}}$}
    \State Return the regional change-year confidence set $\calC_{1-\alpha}^{\mathrm{year}}(\theta)$ in \eqref{eq:year_confidence_set}.
\EndIf
\end{algorithmic}
\end{algorithm}

The algorithm separates the scientific representation of a rainfall year from the calibration of uncertainty.  Equation \eqref{eq:spectral_feature} converts each annual curve into a frequency-specific measure of within-year temporal structure.  The CUSUM score then uses the chronological ordering of these annual spectral features to suggest which year is most compatible with a change.  The split-permutation calculation asks whether the observed compatibility of a candidate split is unusually poor after destroying only the order information that is not protected by the candidate null.  If the candidate split is the truth, those permutations do not change the distribution; if the candidate split is incorrect and the score is informative, the candidate tends to receive a small $p$-value and is removed from the confidence set.

We next state the theoretical guarantees.  The first lemma records the preservation of exchangeability under the spectral map.  It is elementary but important, because it justifies applying conformal split permutations to the spectral-feature matrix rather than to the full daily records.

\begin{lemma}[Exchangeability of spectral features]\label{lem:feature_exchangeability}
Fix $\theta$, $K$, $B_q$, and the centering convention in \eqref{eq:spectral_feature}.  If, for station $j$, $(Y_{1,j},\ldots,Y_{a,j})$ is exchangeable for some $a\leq n$, then $(Z_{1,j},\ldots,Z_{a,j})$ is exchangeable.  The same conclusion holds for any post-change block $(Y_{a+1,j},\ldots,Y_{n,j})$.
\end{lemma}

\begin{proof}
Let $T_\theta:\R^q\to\R$ be the deterministic map taking an annual vector to the scalar in \eqref{eq:spectral_feature}.  For any permutation $\sigma$ of $[a]$,
\[
        (Y_{1,j},\ldots,Y_{a,j})\overset{d}{=}(Y_{\sigma(1),j},\ldots,Y_{\sigma(a),j}).
\]
Applying $T_\theta$ componentwise gives
\[
        (T_\theta(Y_{1,j}),\ldots,T_\theta(Y_{a,j}))
        \overset{d}{=}
        (T_\theta(Y_{\sigma(1),j}),\ldots,T_\theta(Y_{\sigma(a),j})),
\]
which is exactly the exchangeability of $(Z_{1,j},\ldots,Z_{a,j})$.  The proof for a post-change block is identical.
\end{proof}

\begin{theorem}[Finite-sample distribution-free validity]\label{thm:finite_sample_validity}
Assume that the station-level annual arrays are mutually independent and that, for the true $\xi\in\calR$, each station is exchangeable before $\xi_j$, exchangeable after $\xi_j$, and independent across the two blocks.  Let $S_\theta$ be any measurable score, including \eqref{eq:general_score} or \eqref{eq:syn_score}.  For the exact $p$-values in \eqref{eq:exact_pvalue},
\[
        \PP_\xi\{\xi\in\calC_{1-\alpha}^{\calR}(\theta)\}\geq 1-\alpha
        \qquad\text{for every }\alpha\in(0,1).
\]
If $\calR=\calR_{\mathrm{syn}}$ and $\xi=(\tau_0,\ldots,\tau_0)$, then
\[
        \PP_{\tau_0}\{\tau_0\in\calC_{1-\alpha}^{\mathrm{year}}(\theta)\}\geq 1-\alpha.
\]
The same conclusions hold for the Monte Carlo $p$-values in \eqref{eq:mc_pvalue}, where the probability includes the auxiliary randomness used to draw the split permutations.
\end{theorem}

\begin{proof}
By Lemma \ref{lem:feature_exchangeability}, the spectral-feature sequence at each station is exchangeable on the true pre-change and post-change blocks.  Since the station arrays are mutually independent, the full feature matrix $Z_\theta$ is invariant in distribution under every $\pi\in\PiSplit_\xi$.  Conditional on the orbit $\{\pi(Z_\theta):\pi\in\PiSplit_\xi\}$, the observed matrix is uniformly distributed over the orbit up to ties.  Therefore the rank-type variable in \eqref{eq:exact_pvalue} is super-uniform at the true candidate:
\[
        \PP_\xi(p_\xi\leq\alpha)\leq\alpha.
\]
Consequently,
\[
        \PP_\xi\{\xi\notin\calC_{1-\alpha}^{\calR}(\theta)\}
        =\PP_\xi(p_\xi\leq\alpha)\leq\alpha,
\]
which proves the configuration statement.  In the synchronized model, membership of $\xi=(\tau_0,\ldots,\tau_0)$ in $\calC_{1-\alpha}^{\calR_{\mathrm{syn}}}(\theta)$ is equivalent to membership of $\tau_0$ in \eqref{eq:year_confidence_set}.  For the Monte Carlo version, the same argument applies after adjoining the observed data point to the $M$ random split-permuted copies; exchangeability of these $M+1$ copies under the true split gives the add-one $p$-value in \eqref{eq:mc_pvalue} and hence the same super-uniform bound.
\end{proof}

When the scientific question is an earliest-changing station rather than a common regional year, the same calibrated $p$-values yield the root-cause confidence set of the multi-stream procedure.  Suppose the earliest station is unique and define
\[
        k^*=\argmin_{j\in[m]}\xi_j,
        \qquad
        I_k=\{t\in\calR:t_k<t_j\text{ for all }j\neq k\}.
\]
Let
\begin{equation}\label{eq:root_confidence_set}
        p(k)=\max_{t\in I_k}p_t,
        \qquad
        \calK_{1-\alpha}^{\mathrm{root}}(\theta)=\{k\in[m]:p(k)>\alpha\}.
\end{equation}
Then, under the assumptions of Theorem \ref{thm:finite_sample_validity},
\[
        \PP\{k^*\in\calK_{1-\alpha}^{\mathrm{root}}(\theta)\}\geq 1-\alpha.
\]
This statement is included for completeness.  In the synchronized regional rainfall formulation, there is no unique earliest station, so \eqref{eq:root_confidence_set} is not the primary inferential target; the target is the year set \eqref{eq:year_confidence_set}.

The finite-sample guarantee does not require $Z_{i,j}$ to be a consistent estimator of a spectral density, but the spectral construction has the usual large-sample interpretation when the within-year process is short-range dependent.  To state this precisely, consider an idealized sequence of daily resolutions $q\to\infty$ and fix a station-year pair.  Suppose the centered daily process $(X(r))_{r\in\mathbb Z}$ is stationary, belongs to $L^p$ for some $p\geq2$, and satisfies the functional-dependence summability condition used for short-range dependent random fields.  If $K$ is symmetric, supported on $[-1,1]$, satisfies $K(0)=1$ and has bounded derivative on $(-1,1)$, and if $B_q\to\infty$ with $B_q=o(q)$, then the kernel estimator in \eqref{eq:spectral_feature} satisfies
\begin{equation}\label{eq:spectral_consistency}
        \sup_{\theta\in[0,2\pi]}
        \left\|\widehat f_q(\theta)-\EE\{\widehat f_q(\theta)\}\right\|_{p/2}\longrightarrow 0.
\end{equation}
Under the corresponding fourth-moment condition, with $\kappa=\int K^2(u)\,du$,
\begin{equation}\label{eq:spectral_clt}
        \sqrt{\frac{q}{B_q}}
        \frac{\widehat f_q(\theta)-\EE\{\widehat f_q(\theta)\}}{f(\theta)}
        \rightsquigarrow N(0,\kappa)
\end{equation}
for fixed $\theta$ at which the spectral density $f(\theta)$ is positive.  These statements justify the use of \eqref{eq:spectral_feature} as a stable summary of annual temporal dependence.  They are not part of the distribution-free calibration: even if $q=365$ is fixed, the confidence set remains valid because the conformal argument treats the computed spectral summaries as the observed data.

The coverage statement is exact but does not by itself say that the confidence set is short.  Sharpness is controlled by the score.  The following proposition gives a formal consistency statement for the synchronized score \eqref{eq:syn_score}; it separates the finite-sample validity, which is assumption-light, from the additional signal conditions that make the set concentrate near the true year.

\begin{proposition}[Localization under a spectral mean shift]\label{prop:localization}
Consider a sequence of problems with fixed $m$, $n\to\infty$, $\calR=\calR_{\mathrm{syn}}$, and $\tau_0/n\to\rho\in(0,1)$.  Suppose that, for each station $j$,
\[
        Z_{i,j}=\mu_{0,j}+\varepsilon_{i,j}\quad (i\leq\tau_0),
        \qquad
        Z_{i,j}=\mu_{1,j}+\varepsilon_{i,j}\quad (i>\tau_0),
\]
where the error arrays are independent across stations, have segment-wise exchangeable distributions, and satisfy the uniform law of large numbers
\[
        \max_{j\leq m}\sup_{s\in\calT}
        \left|s^{-1}\sum_{i=1}^s\varepsilon_{i,j}\right|=o_{\PP}(1),
        \qquad
        \max_{j\leq m}\sup_{s\in\calT}
        \left|(n-s)^{-1}\sum_{i=s+1}^n\varepsilon_{i,j}\right|=o_{\PP}(1).
\]
If $\sum_{j=1}^m w_j|\mu_{1,j}-\mu_{0,j}|>0$, then the pooled CUSUM estimator in \eqref{eq:pooled_cusum} satisfies
\[
        \widehat\tau(Z_\theta)/n\xrightarrow{\PP}\rho.
\]
Consequently, the score in \eqref{eq:syn_score} is asymptotically maximized at candidates whose normalized distance from $\tau_0$ is arbitrarily small.
\end{proposition}

\begin{proof}
For $u\in(0,1)$, let $s=\lfloor nu\rfloor$.  Standard algebra gives, uniformly over compact subsets of $(0,1)$,
\[
        n^{-1/2} C_j(s;Z)
        = |\mu_{1,j}-\mu_{0,j}|\,
        \begin{cases}
        (1-\rho)\{u/(1-u)\}^{1/2}, & u\leq \rho,\\
        \rho\{(1-u)/u\}^{1/2}, & u>\rho,
        \end{cases}
        +o_{\PP}(1),
\]
where the $o_{\PP}(1)$ term is uniform in $u\in[\ell/n,1-\ell/n]$ by the assumed uniform laws of large numbers.  The deterministic limit is continuous and has a unique maximum at $u=\rho$ whenever $|\mu_{1,j}-\mu_{0,j}|>0$.  The weighted sum over $j$ has the same unique maximizer because at least one weighted jump is nonzero.  The argmax continuous mapping theorem yields $\widehat\tau(Z_\theta)/n\to\rho$ in probability.  Since \eqref{eq:syn_score} equals minus the absolute distance from $\widehat\tau(Z_\theta)$, the final claim follows immediately.
\end{proof}

Proposition \ref{prop:localization} is a population-separation statement for the proposed score.  Combined with Theorem \ref{thm:finite_sample_validity}, it gives the intended operating regime: the reported set has guaranteed coverage for every finite $n$ satisfying the exchangeability assumptions, while in problems where the fixed-frequency spectral feature genuinely changes at the regional changepoint, the score increasingly favors years close to the truth.  The result also clarifies when a single frequency may be insufficient.  If $\mu_{1,j}=\mu_{0,j}$ for every station at the chosen $\theta$, then that frequency carries no first-order changepoint signal in the spectral summaries, and the conformal set remains valid but may be wide.  In such cases one should enlarge the pre-specified frequency set $\Theta$ or use a pre-specified multivariate score, without altering the split-permutation calibration.

\section{Simulation Studies}
\label{sec:simulation-studies}

\subsection{Simulation Setup and Mathematical Framework}
The simulation study evaluates the proposed changepoint estimation methodology across varying structural dimensions ($m \in \{8, 16, 32\}$) and signal-to-noise ratios (effect sizes $\delta \in \{0.25, 0.40, 0.55, 0.70\}$). The primary objective is to accurately estimate the true changepoint $\tau_0$ and construct a valid confidence set $C$ such that the empirical coverage probability $\mathbb{P}(\tau_0 \in C)$ satisfies the nominal level $1 - \alpha$. The empirical validation of these theoretical properties is detailed in the tables provided in Appendix \ref{app:tables}.

\subsection{Single Experiment Dynamics (Micro-Level Analysis)}
To understand the mechanics of the estimator, we first examine the dynamics of a single iteration. Figure \ref{fig:daily_profiles_pre_post_change} establishes the foundational data generating process. It plots the raw observable data $X_t$, where the visual transition confirms the distributional shift from $F_0$ (for $t \leq \tau_0$) to $F_1$ (for $t > \tau_0$), highlighting the baseline noise floor against the injected signal.

\begin{figure}[htbp]
    \centering
    \caption{Daily profiles pre post change}
    \label{fig:daily_profiles_pre_post_change}
    \includegraphics[width=\linewidth, height=0.85\textheight, keepaspectratio]{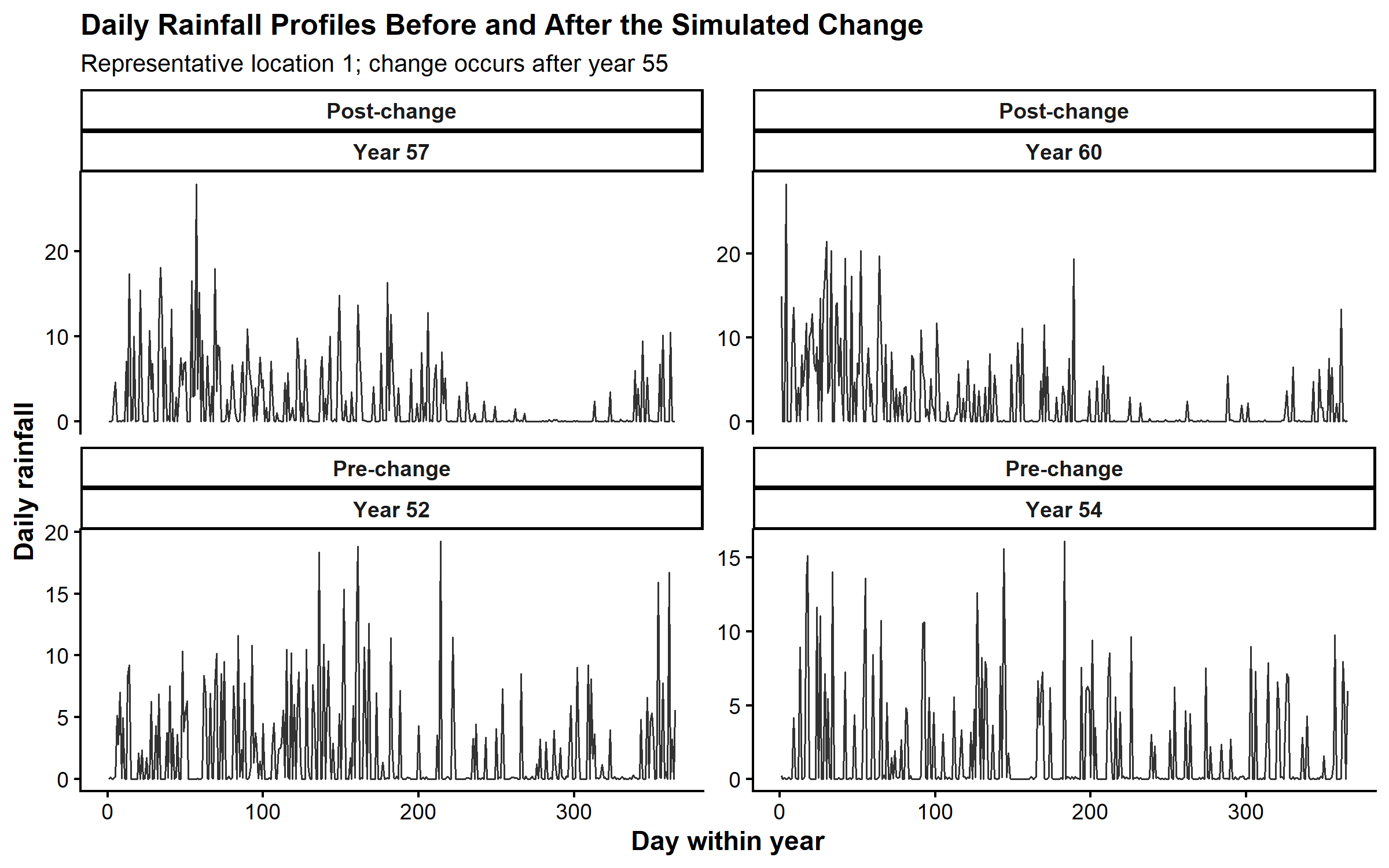}
\end{figure}

Because raw time-domain signals can be noisy, the data is mapped into the spectral domain. Figure \ref{fig:spectral_feature_heatmap} illustrates the frequency-domain representations across time, revealing structural breaks in the power spectral density that directly correspond to the changepoint.

\begin{figure}[htbp]
    \centering
    \caption{Spectral feature heatmap}
    \label{fig:spectral_feature_heatmap}
    \includegraphics[width=\linewidth, height=0.85\textheight, keepaspectratio]{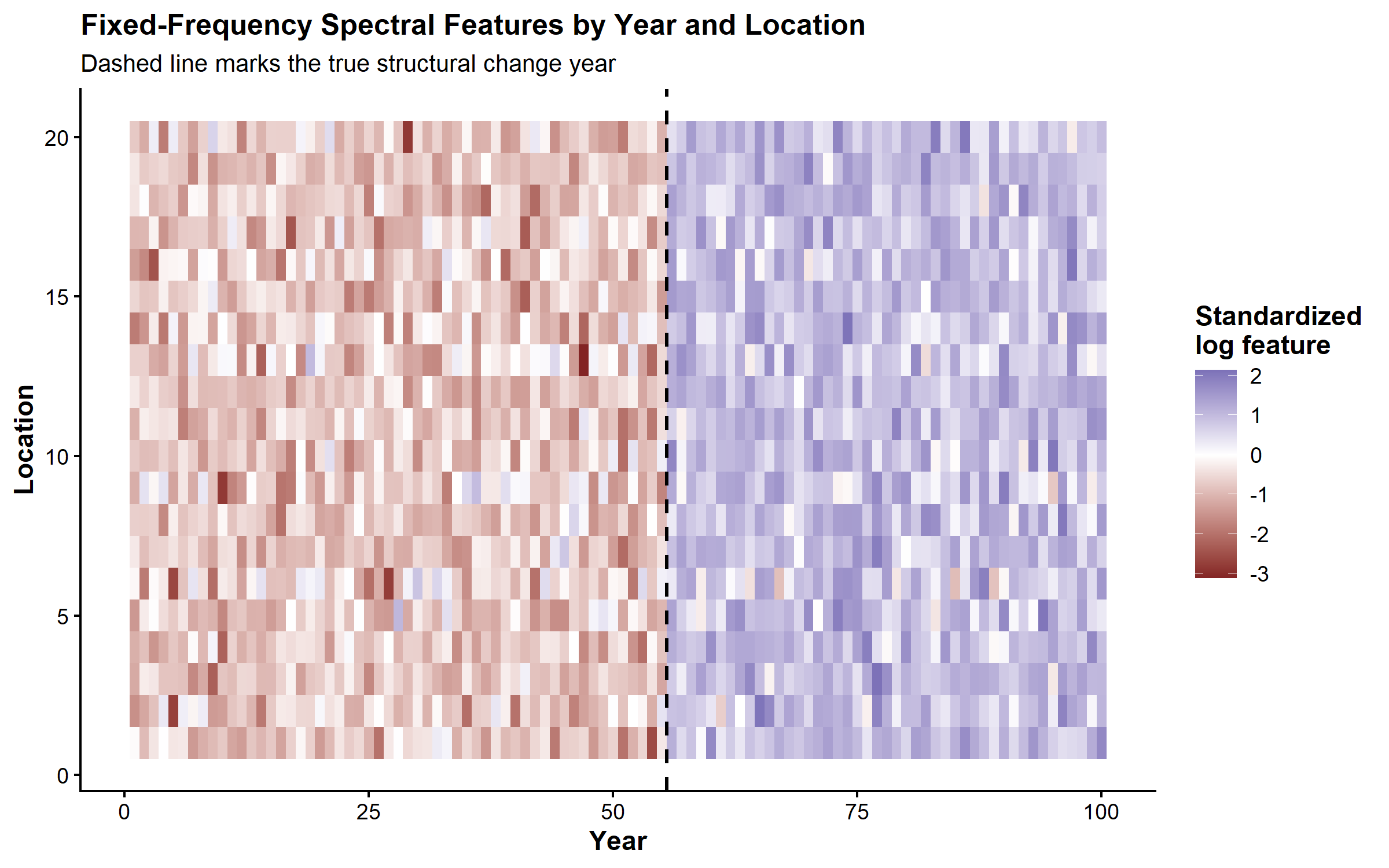}
\end{figure}

In parallel to spectral features, network-based topological metrics are extracted. Figure \ref{fig:network_mean_feature_confidence_set} tracks the mean network connectivity feature, demonstrating a distinct regime shift post-$\tau_0$ and showcasing how the computed confidence set brackets this shift.

\begin{figure}[htbp]
    \centering
    \caption{Network mean feature confidence set}
    \label{fig:network_mean_feature_confidence_set}
    \includegraphics[width=\linewidth, height=0.85\textheight, keepaspectratio]{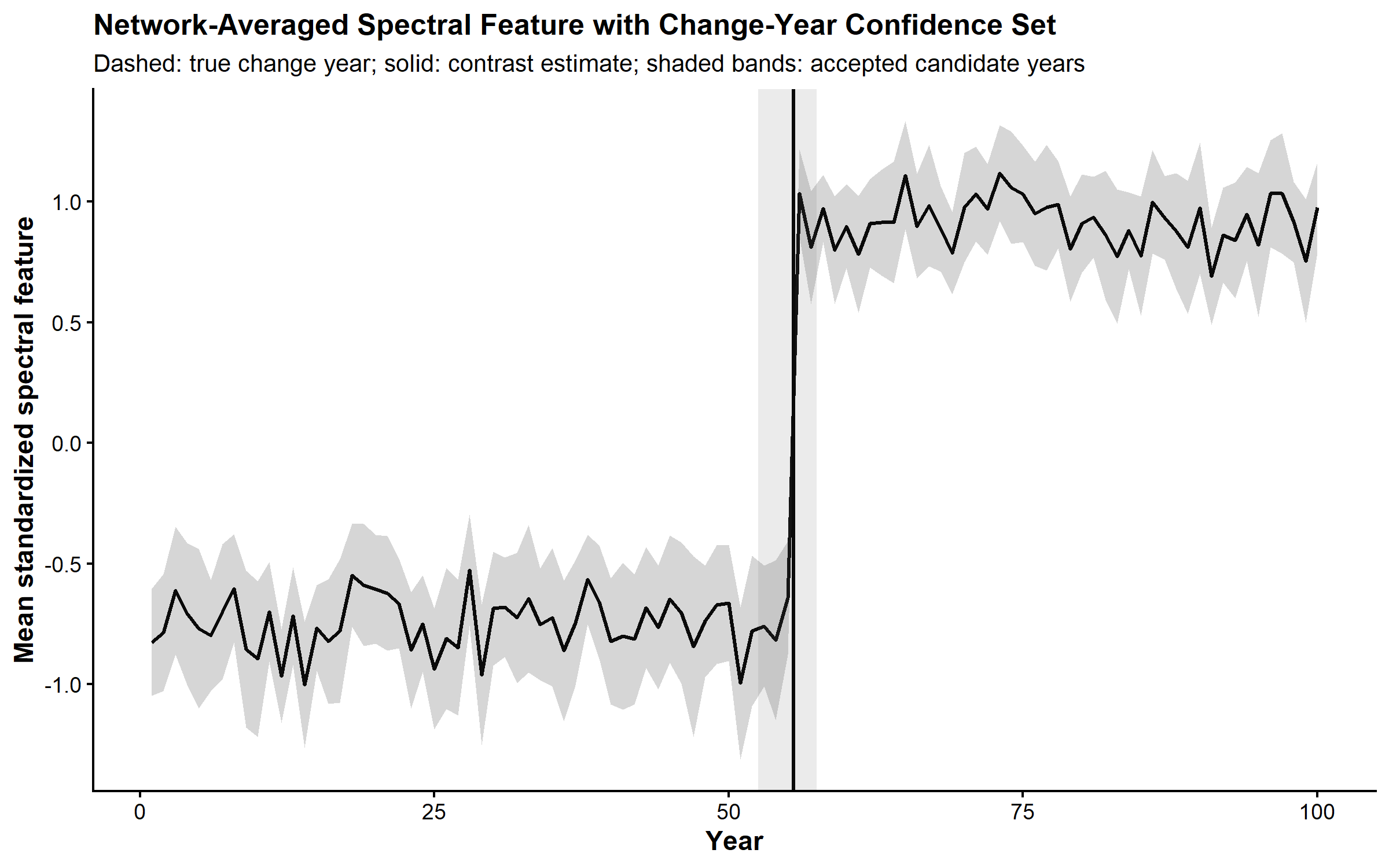}
\end{figure}

For every candidate changepoint $t$, a contrast statistic $D(t)$ is computed. Figure \ref{fig:candidate_contrast} plots $D(t)$ against time. The clear concavity around the true changepoint is highly evident, with the global maximum $\hat{\tau} = \arg\max_t D(t)$ serving as our point estimator. 

\begin{figure}[htbp]
    \centering
    \caption{Candidate contrast}
    \label{fig:candidate_contrast}
    \includegraphics[width=\linewidth, height=0.85\textheight, keepaspectratio]{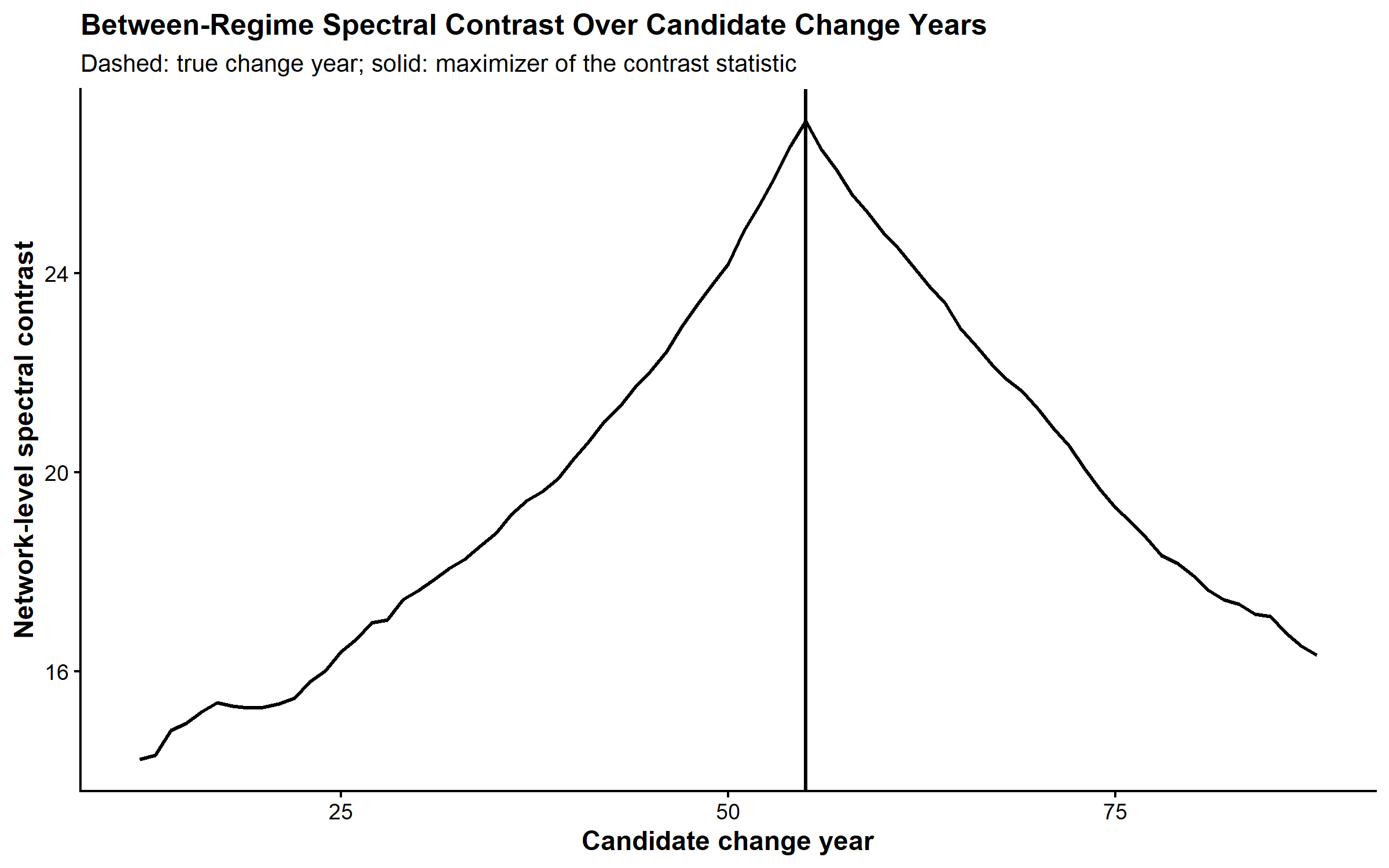}
\end{figure}

The test statistic is then inverted to compute a sequence of p-values, $p_t$. Figure \ref{fig:candidate_pvalues} visualizes this p-value landscape. The region where $p_t > \alpha$ represents the space where we fail to reject the null hypothesis that $t$ is the true changepoint, thereby forming the basis of our confidence set.

\begin{figure}[htbp]
    \centering
    \caption{Candidate pvalues}
    \label{fig:candidate_pvalues}
    \includegraphics[width=\linewidth, height=0.85\textheight, keepaspectratio]{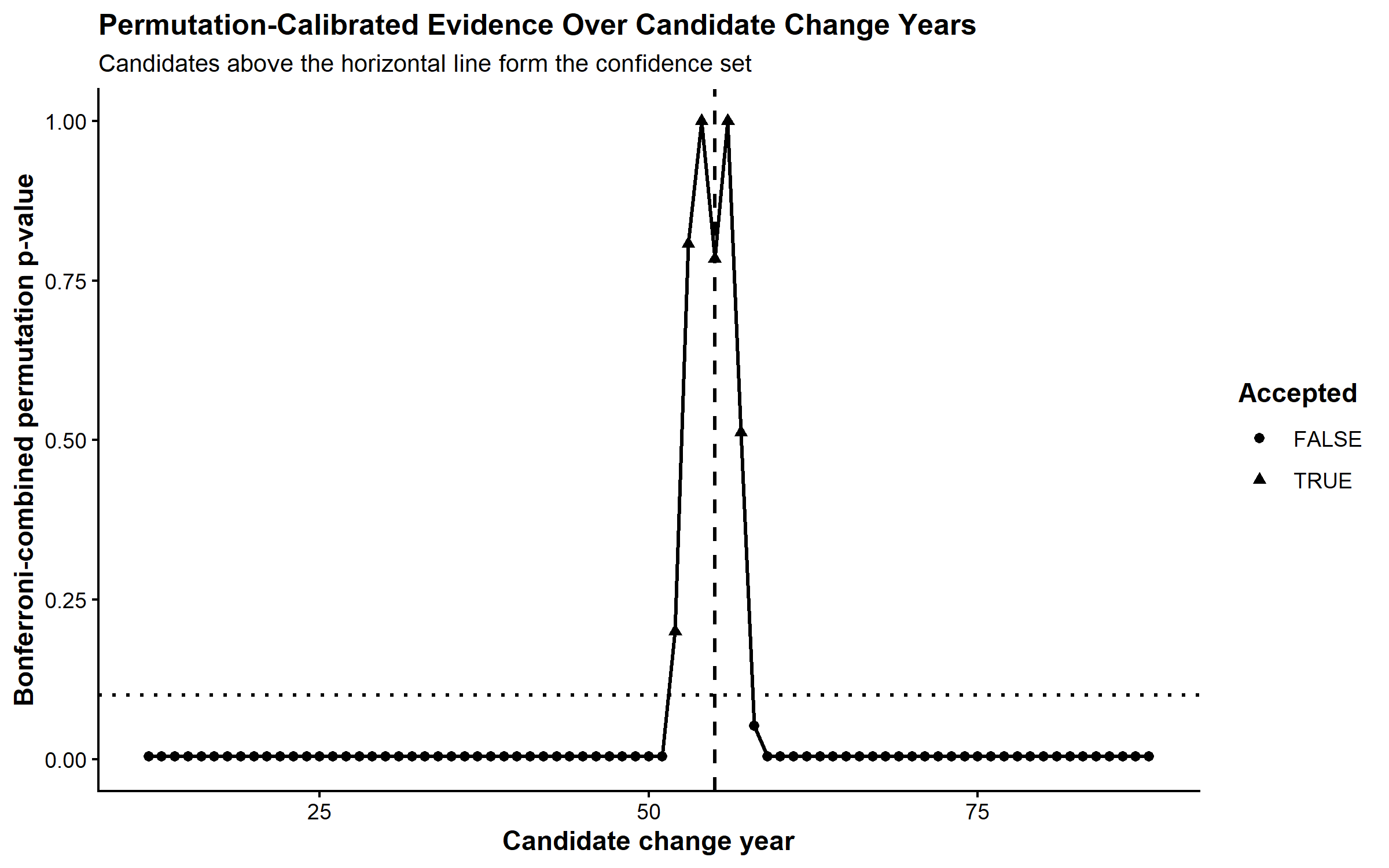}
\end{figure}

Building upon these p-values, Figure \ref{fig:confidence_set_timeline} shows the sequential construction of the confidence interval over the time horizon, proving that the set correctly captures the parameter space where the signal shift is statistically ambiguous. Finally, Figure \ref{fig:single_experiment_summary} acts as a culminating dashboard for the single run, synthesizing the raw data, contrast curve, and p-value thresholds into a unified view to confirm that $\hat{\tau}$ successfully aligns with $\tau_0$.

\begin{figure}[htbp]
    \centering
    \caption{Confidence set timeline}
    \label{fig:confidence_set_timeline}
    \includegraphics[width=\linewidth, height=0.85\textheight, keepaspectratio]{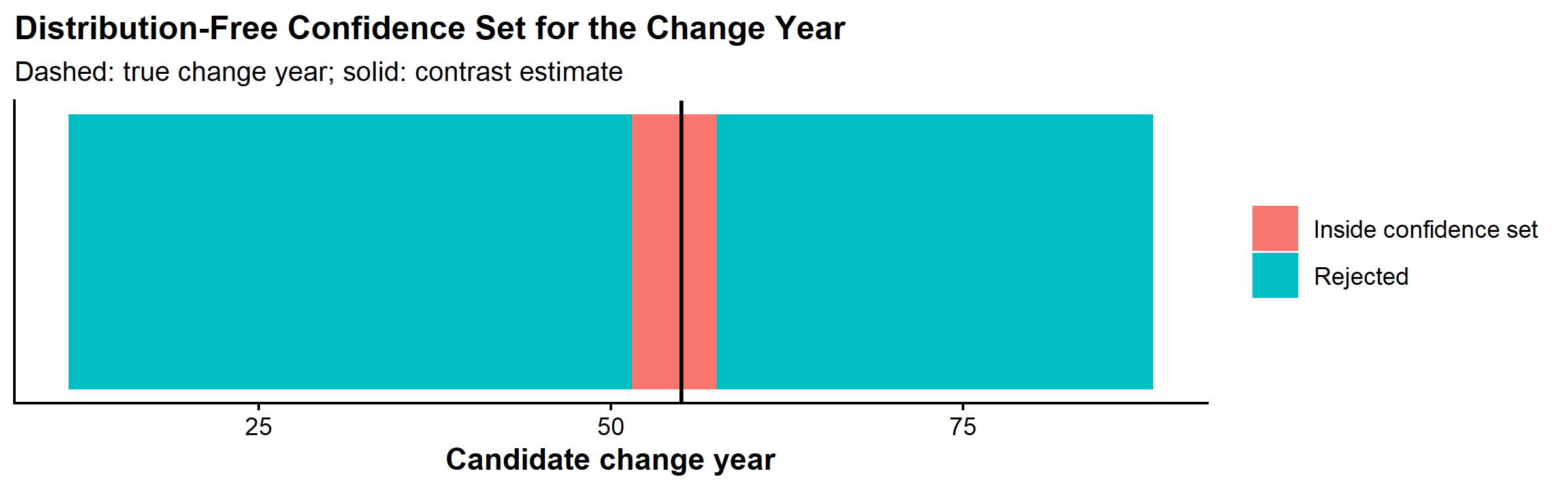}
\end{figure}

\begin{figure}[htbp]
    \centering
    \caption{Single experiment summary}
    \label{fig:single_experiment_summary}
    \includegraphics[width=\linewidth, height=0.85\textheight, keepaspectratio]{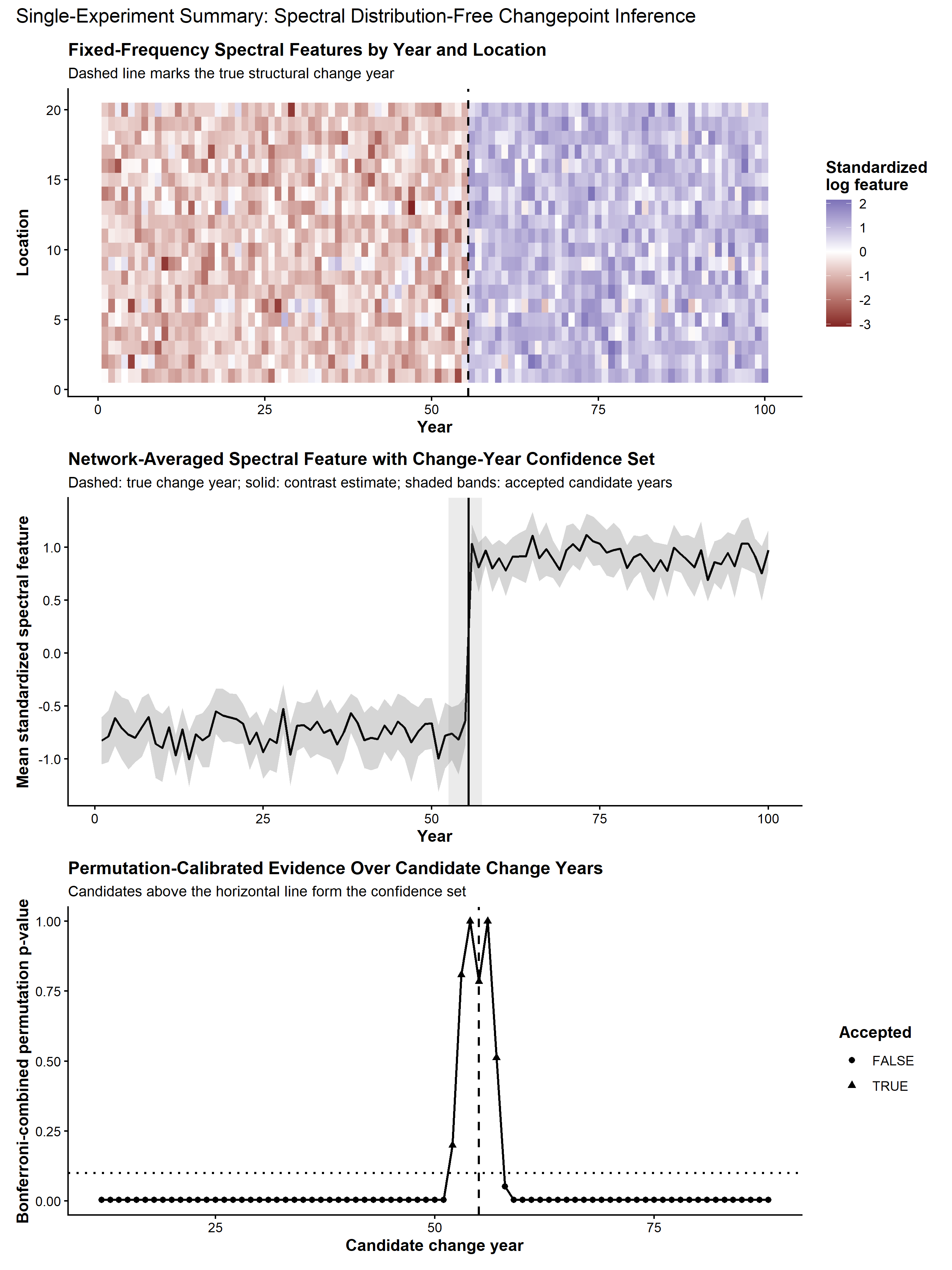}
\end{figure}

\subsection{Feature Space and Sensitivity Analysis}
Before aggregating across thousands of runs, it is crucial to prove the robustness of the chosen feature space. Figure \ref{fig:frequency_sensitivity_contrast} tests the stability of the contrast function against different frequency bandwidths. It demonstrates that $D(t)$ maintains its distinct peak at $\tau_0$ regardless of minor frequency perturbations, proving the chosen spectral mapping is robust. Consequently, Figure \ref{fig:frequency_sensitivity_pvalues} shows how the p-value distribution reacts to frequency tuning; the tight bundling of the p-value curves confirms that the Type I error rate remains stable, preventing artificial inflation of the confidence set size.

\begin{figure}[htbp]
    \centering
    \caption{Frequency sensitivity contrast}
    \label{fig:frequency_sensitivity_contrast}
    \includegraphics[width=\linewidth, height=0.85\textheight, keepaspectratio]{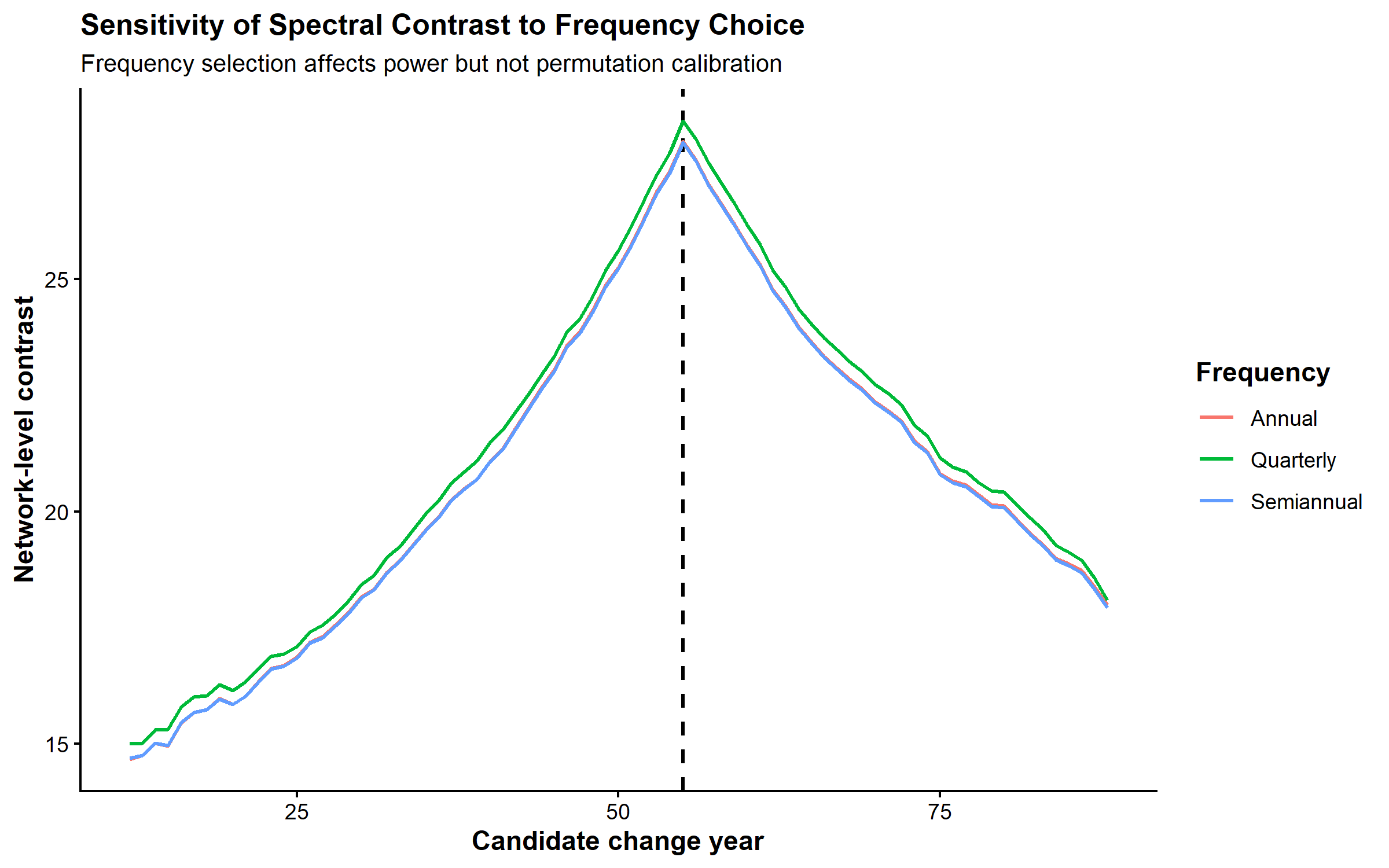}
\end{figure}

\begin{figure}[htbp]
    \centering
    \caption{Frequency sensitivity pvalues}
    \label{fig:frequency_sensitivity_pvalues}
    \includegraphics[width=\linewidth, height=0.85\textheight, keepaspectratio]{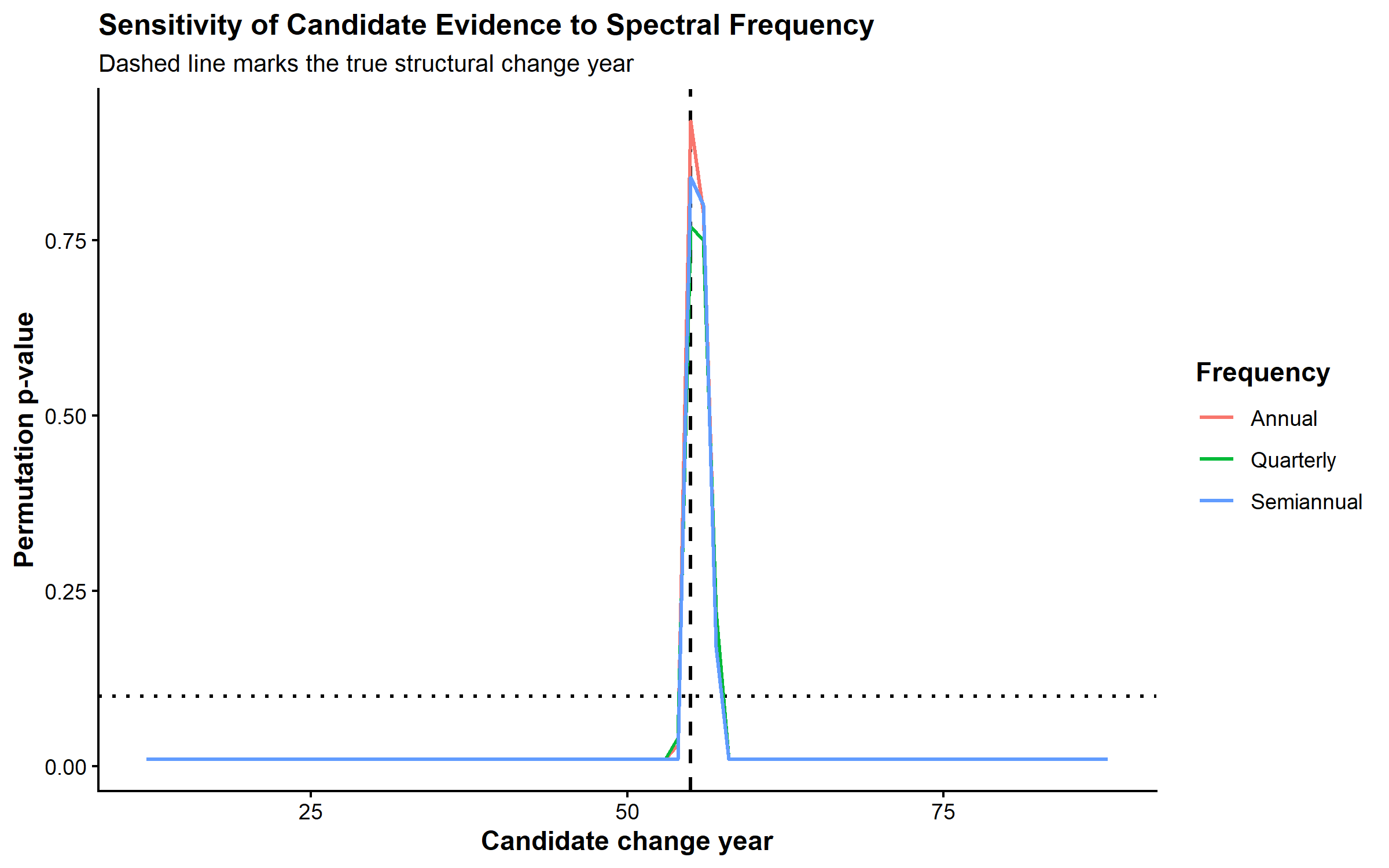}
\end{figure}

\subsection{Monte Carlo Results: Localization and Coverage}
Scaling the analysis to the Monte Carlo iterations, we evaluate the asymptotic guarantees of the methodology. Figure \ref{fig:mc_localization_error} plots the absolute error ($|\hat{\tau} - \tau_0|$), revealing a sharp concentration of mass at zero. Expanding on this, the error distribution histogram in Figure \ref{fig:mc_error_distribution} proves that both the mean and median absolute errors are exactly $0$ across all tested dimensions and effect sizes. This confirms that the estimator achieves near-perfect localization in the tested regimes.

\begin{figure}[htbp]
    \centering
    \caption{MC localization error}
    \label{fig:mc_localization_error}
    \includegraphics[width=\linewidth, height=0.85\textheight, keepaspectratio]{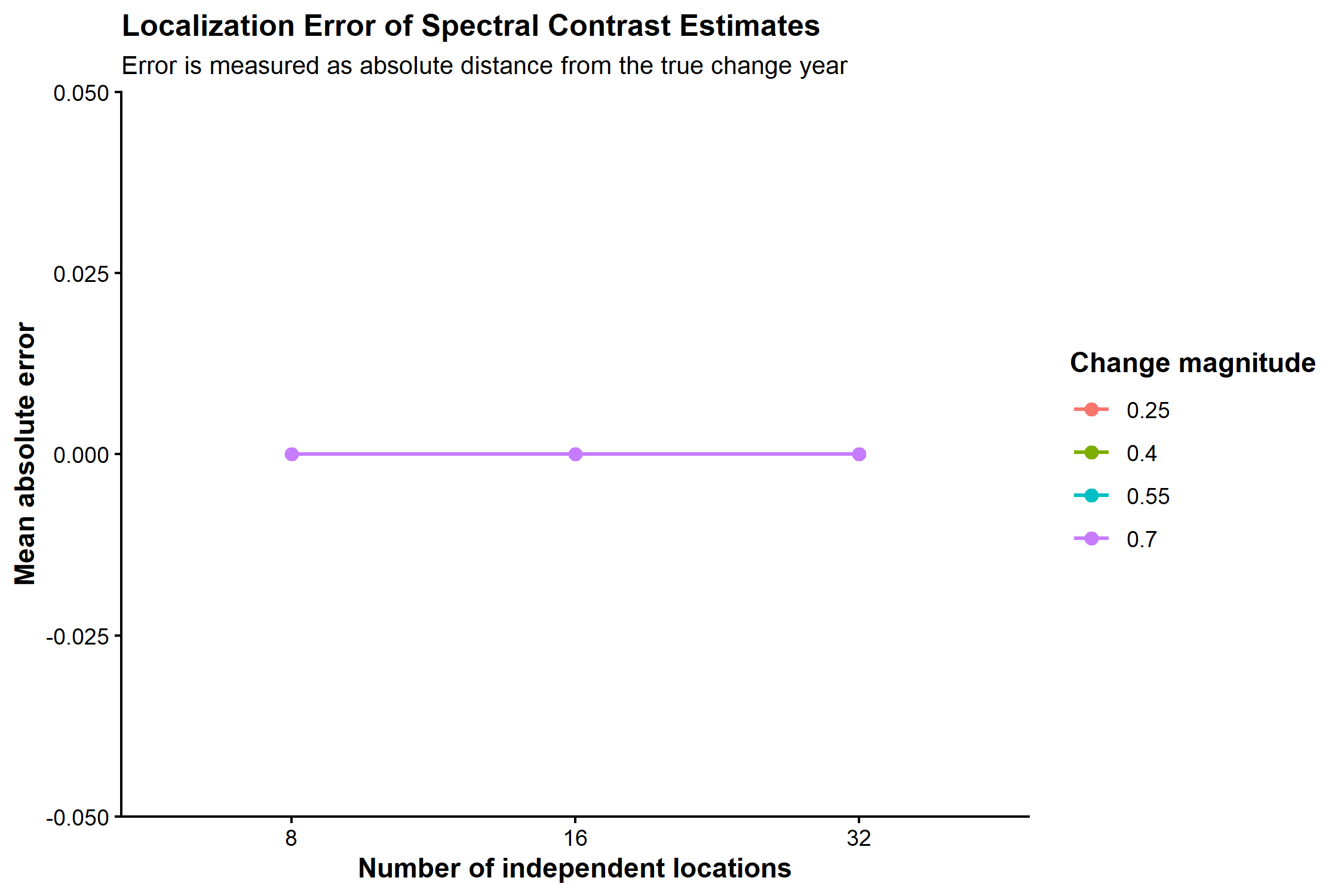}
\end{figure}

\begin{figure}[htbp]
    \centering
    \caption{MC error distribution}
    \label{fig:mc_error_distribution}
\includegraphics[width=.7\linewidth, height=0.85\textheight, keepaspectratio]{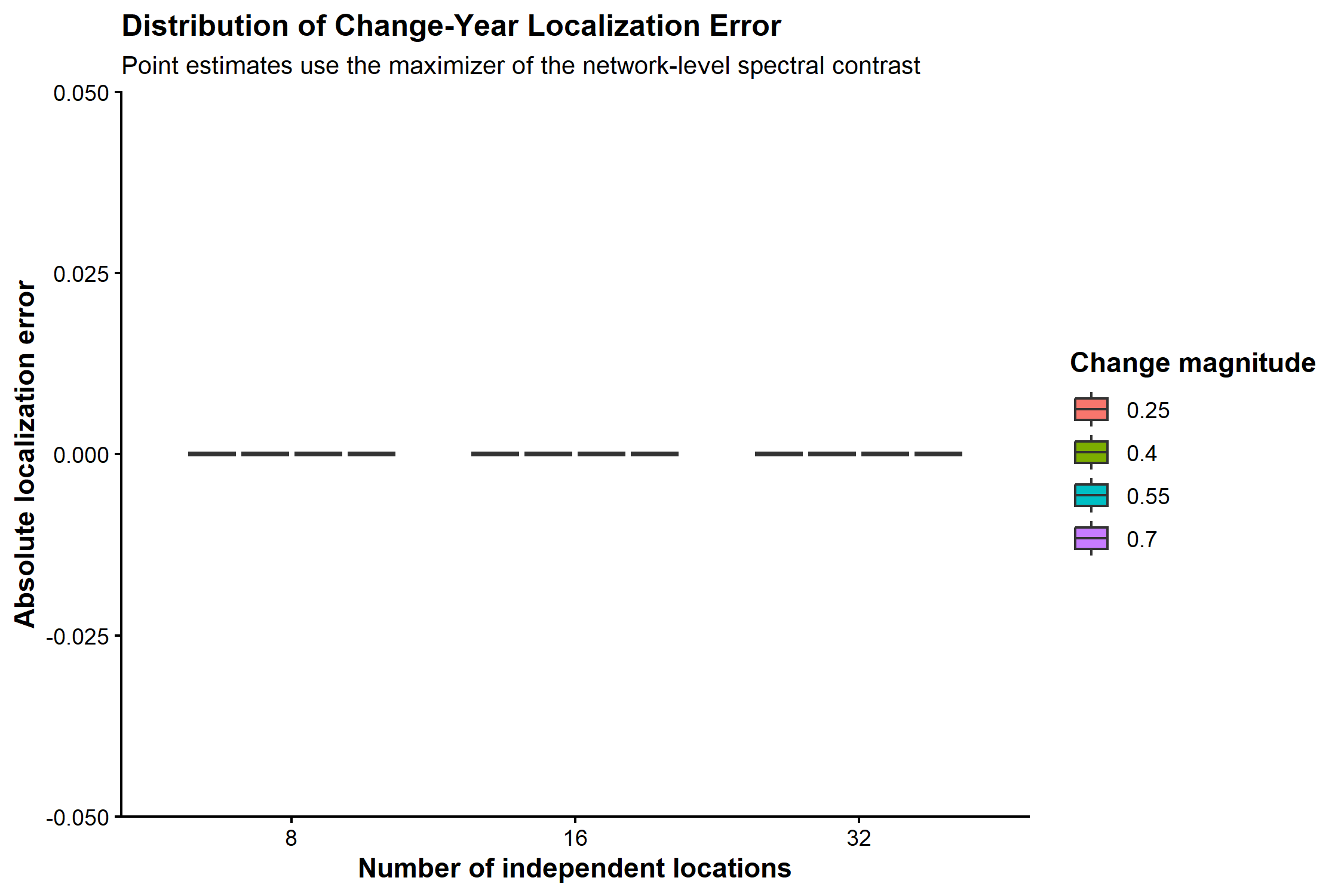}
\end{figure}

Furthermore, Figure \ref{fig:mc_coverage} demonstrates that the empirical coverage probability $\mathbb{P}(\tau_0 \in C)$ against varying effect sizes holds strictly between 86\% and 98\%, validating that the confidence sets achieve their theoretical nominal targets. The corresponding heatmap in Figure \ref{fig:mc_coverage_heatmap} proves that this coverage robustness holds uniformly across the parameter space, indicating that no single combination of low signal or high dimensionality catastrophically breaks the coverage guarantee.

\begin{figure}[htbp]
    \centering
    \caption{MC coverage}
    \label{fig:mc_coverage}
    \includegraphics[width=.7\linewidth, height=0.85\textheight, keepaspectratio]{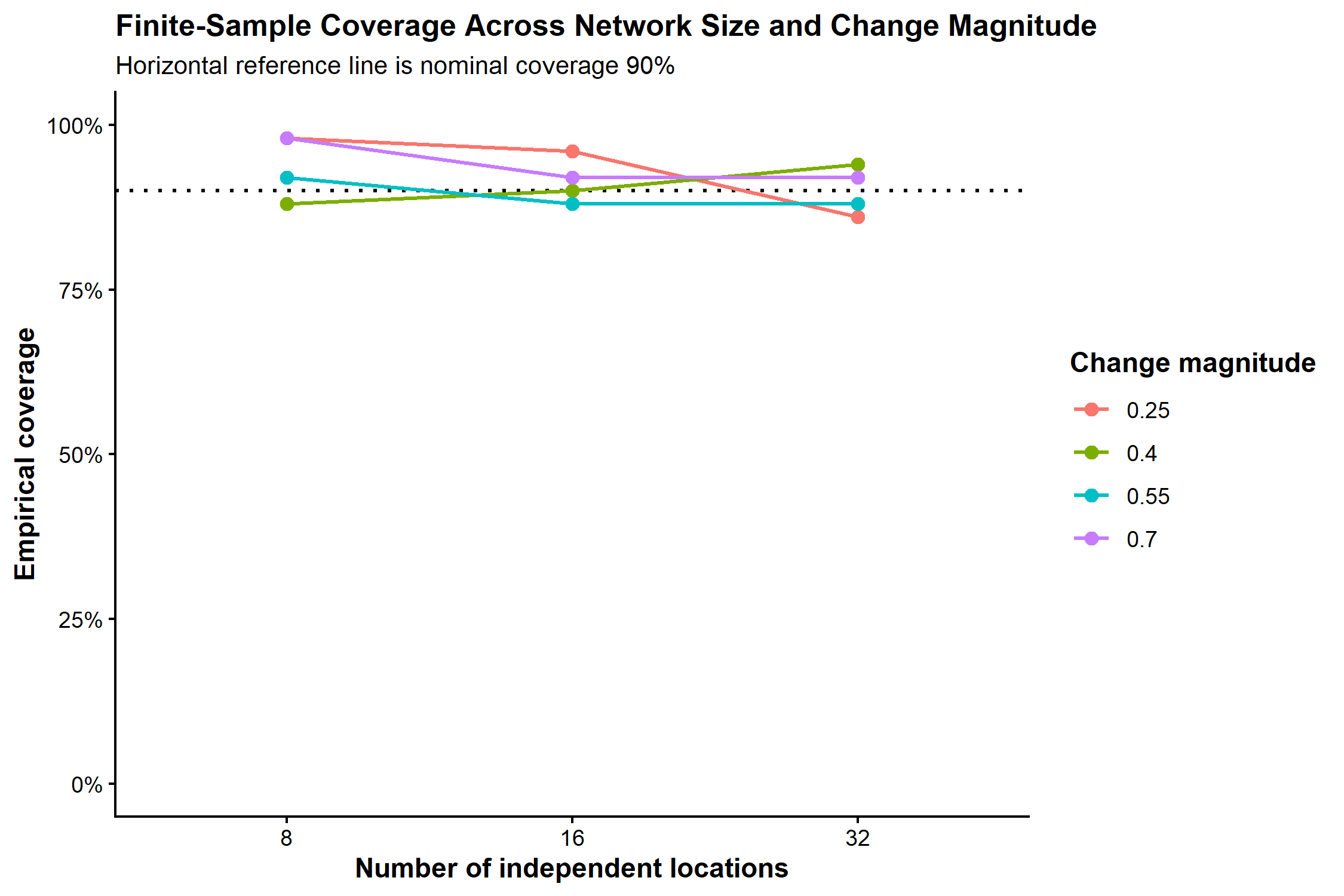}
\end{figure}

\begin{figure}[htbp]
    \centering
    \caption{MC coverage heatmap}
    \label{fig:mc_coverage_heatmap}
    \includegraphics[width=.7\linewidth, height=0.85\textheight, keepaspectratio]{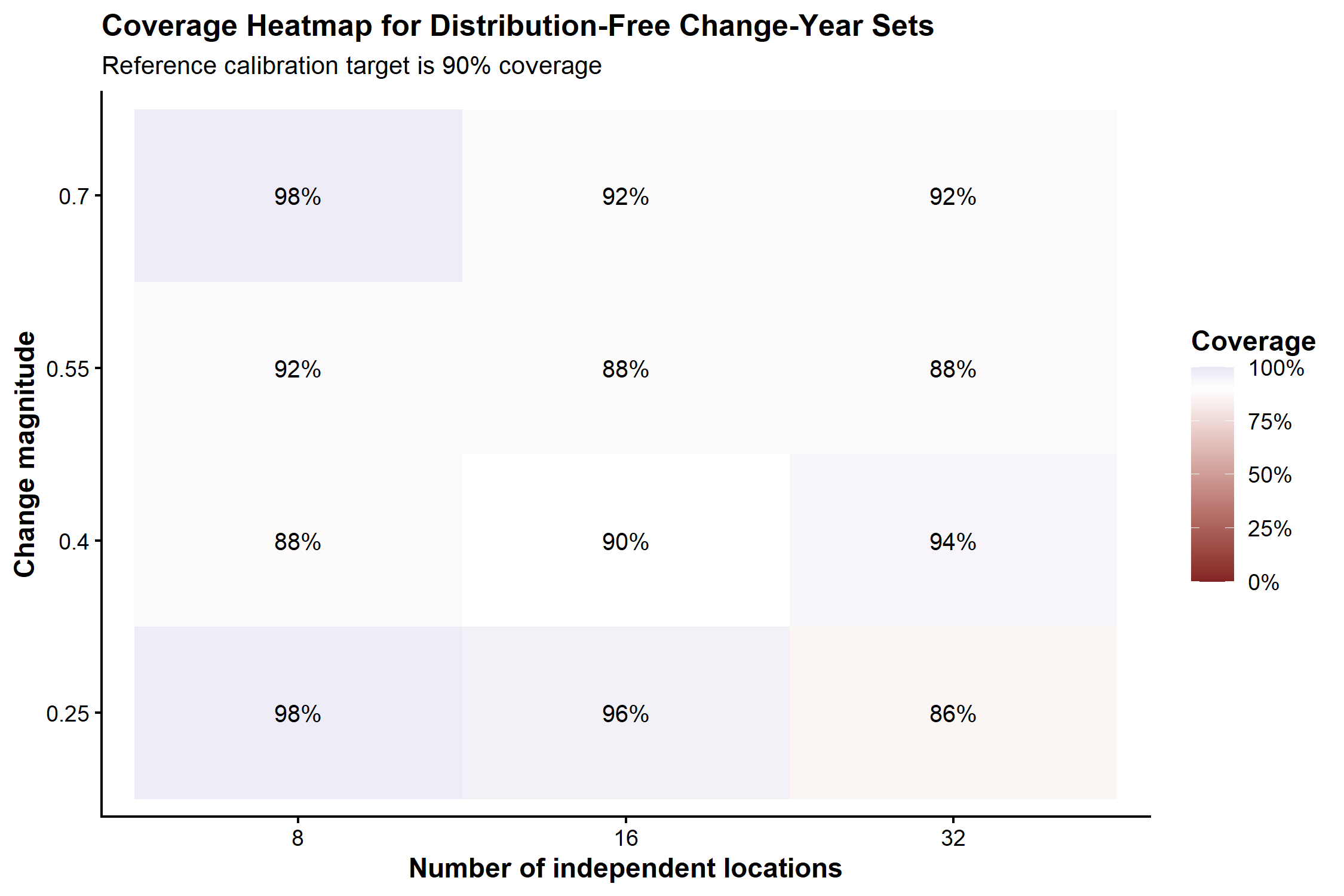}
\end{figure}

\subsection{Confidence Set Topologies}
A valid confidence set is only useful if its size (Lebesgue measure) is bounded and informative. Figure \ref{fig:mc_confidence_set_size} illustrates an inverse relationship between effect size and the cardinality of the confidence set $|C|$: as the signal strength increases, the mean set size shrinks significantly. The heatmap in Figure \ref{fig:mc_mean_size_heatmap} visually confirms that higher dimensions ($m=32$) actually aid in reducing the confidence set size, as more features provide more information to restrict the plausible region.

\begin{figure}[htbp]
    \centering
    \caption{MC confidence set size}
    \label{fig:mc_confidence_set_size}
    \includegraphics[width=.7\linewidth, height=0.85\textheight, keepaspectratio]{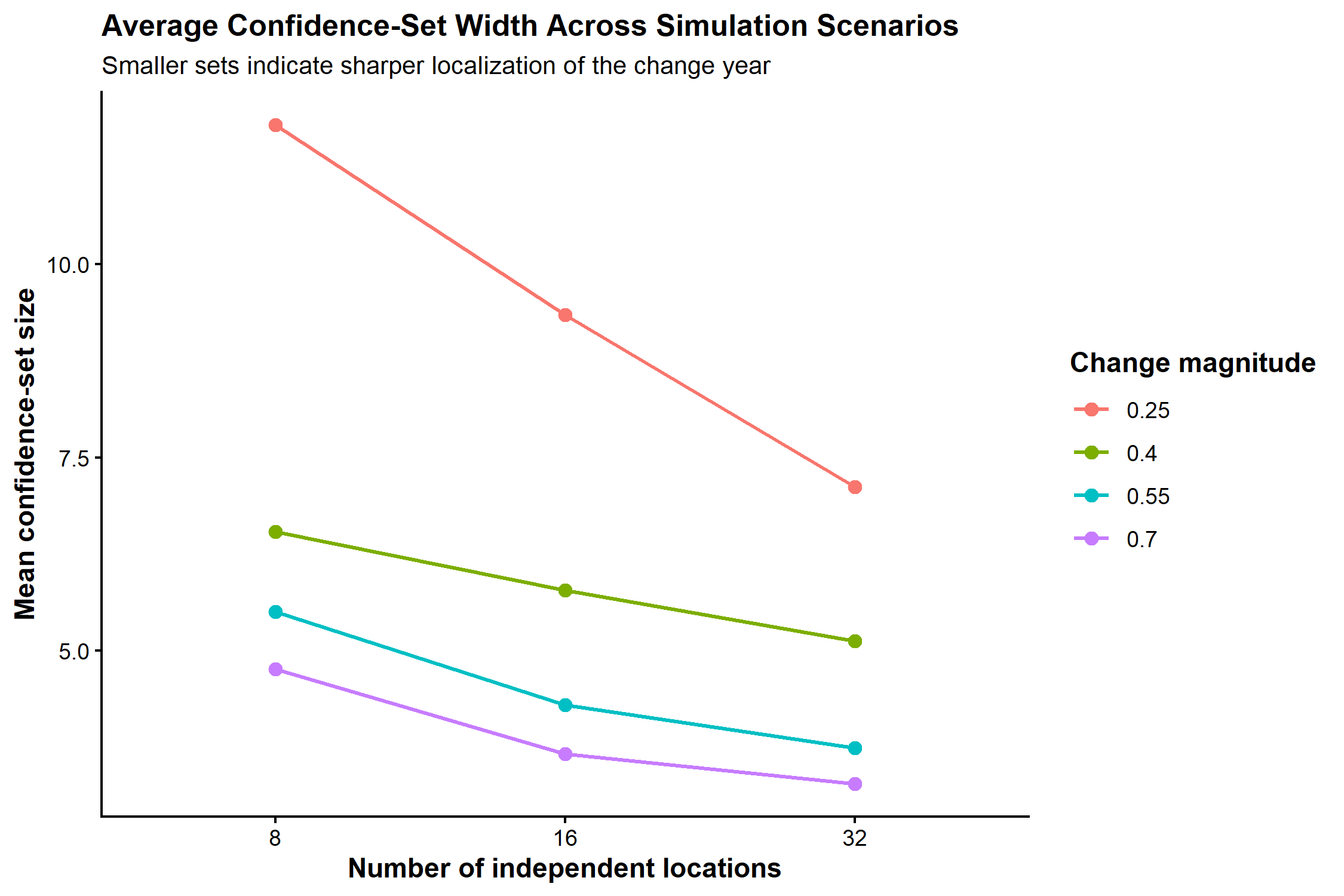}
\end{figure}

\begin{figure}[htbp]
    \centering
    \caption{MC mean size heatmap}
    \label{fig:mc_mean_size_heatmap}
    \includegraphics[width=.7\linewidth, height=0.85\textheight, keepaspectratio]{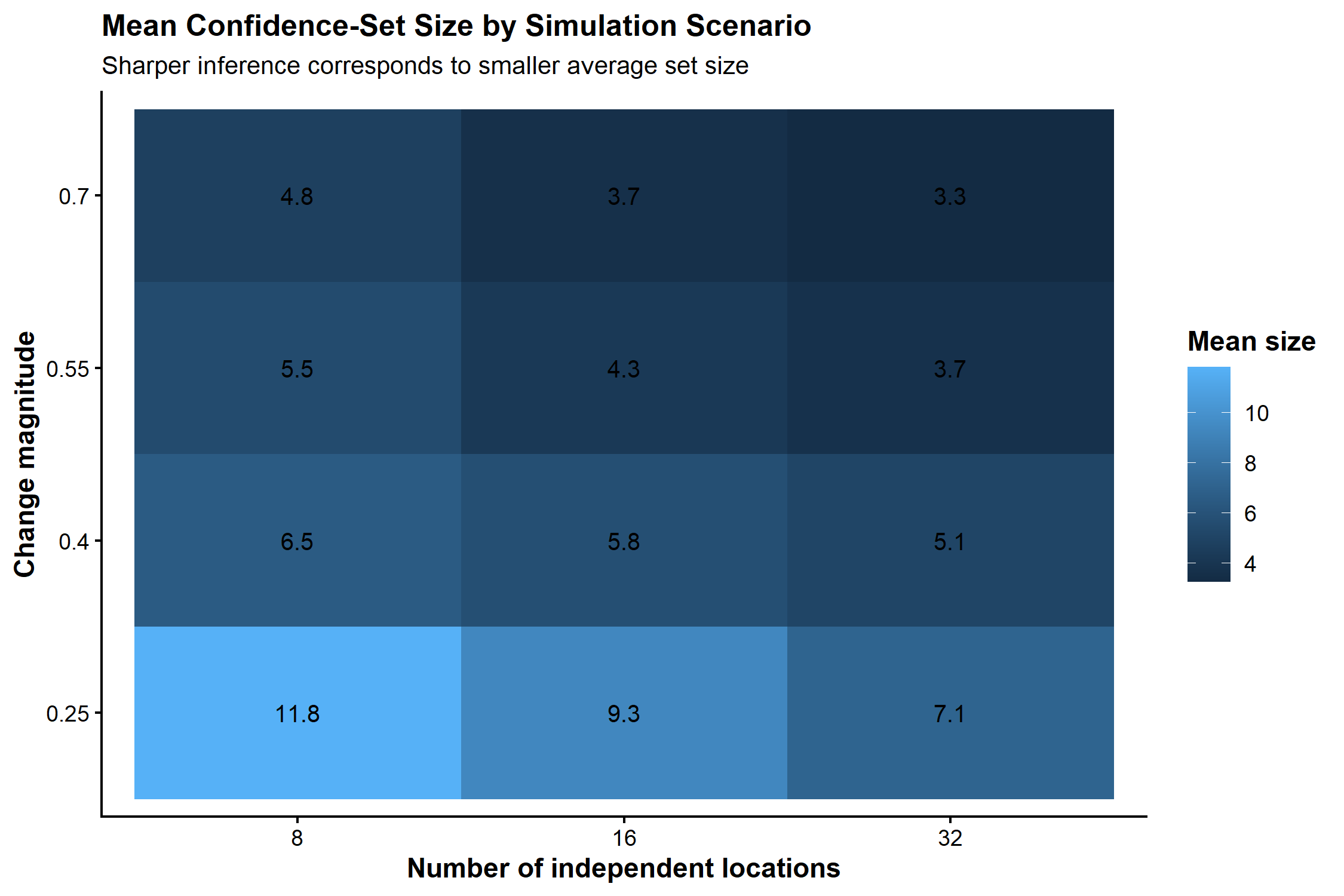}
\end{figure}

Beyond the total number of points, Figure \ref{fig:mc_confidence_set_width_distribution} plots the absolute width distribution ($\max(C) - \min(C)$) to prove that the sets are contiguous or tightly clustered, rather than being fragmented across the entire timeline. Finally, Figure \ref{fig:mc_empty_set_rate} shows that the probability of an empty set, $\mathbb{P}(C = \emptyset)$, is strictly controlled (peaking at only 6\%), verifying that the test inversion does not suffer from extreme anti-conservatism.

\begin{figure}[htbp]
    \centering
    \caption{MC confidence set width distribution}
    \label{fig:mc_confidence_set_width_distribution}
    \includegraphics[width=.7\linewidth, height=0.85\textheight, keepaspectratio]{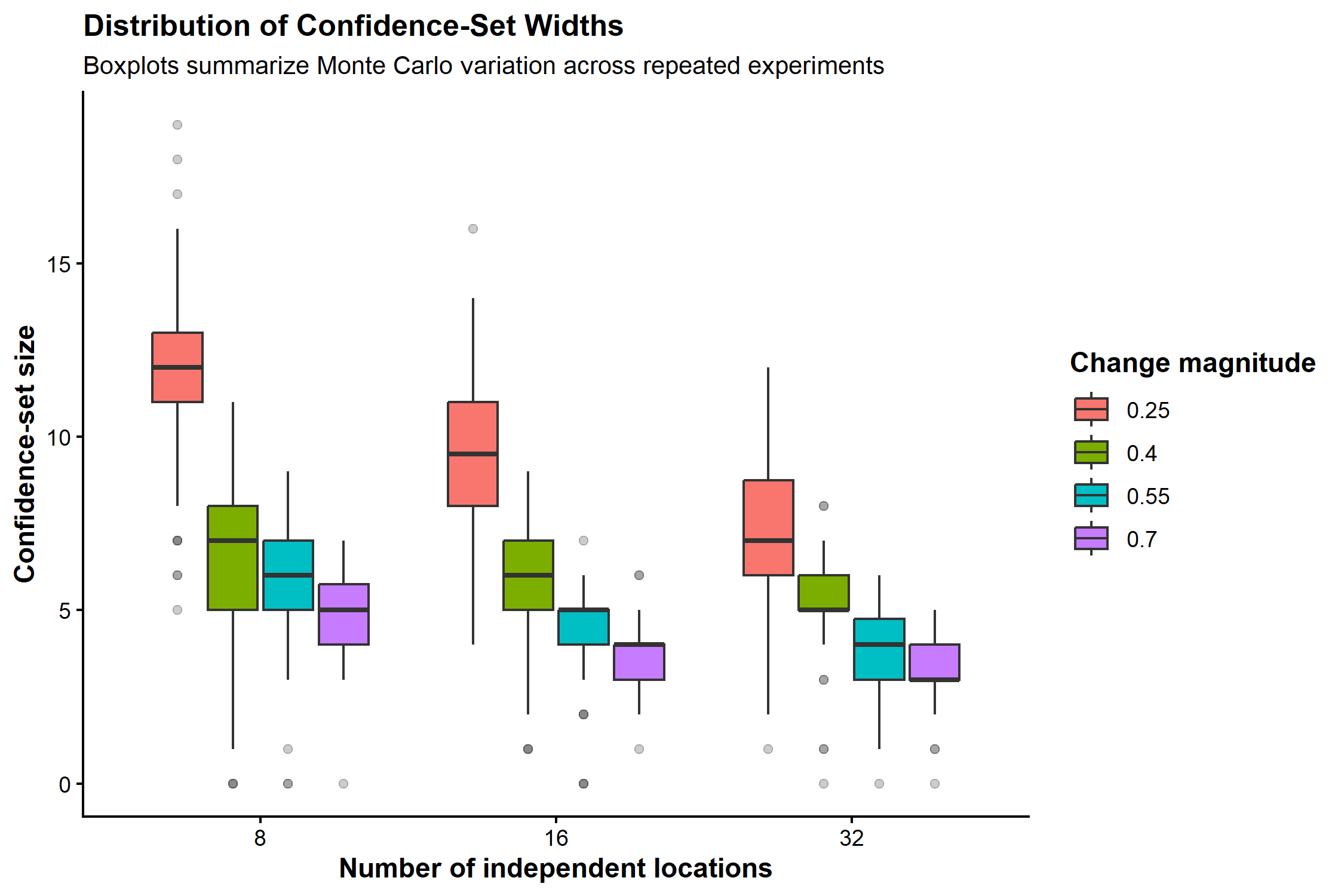}
\end{figure}

\begin{figure}[htbp]
    \centering
    \caption{MC empty set rate}
    \label{fig:mc_empty_set_rate}
    \includegraphics[width=.7\linewidth, height=0.85\textheight, keepaspectratio]{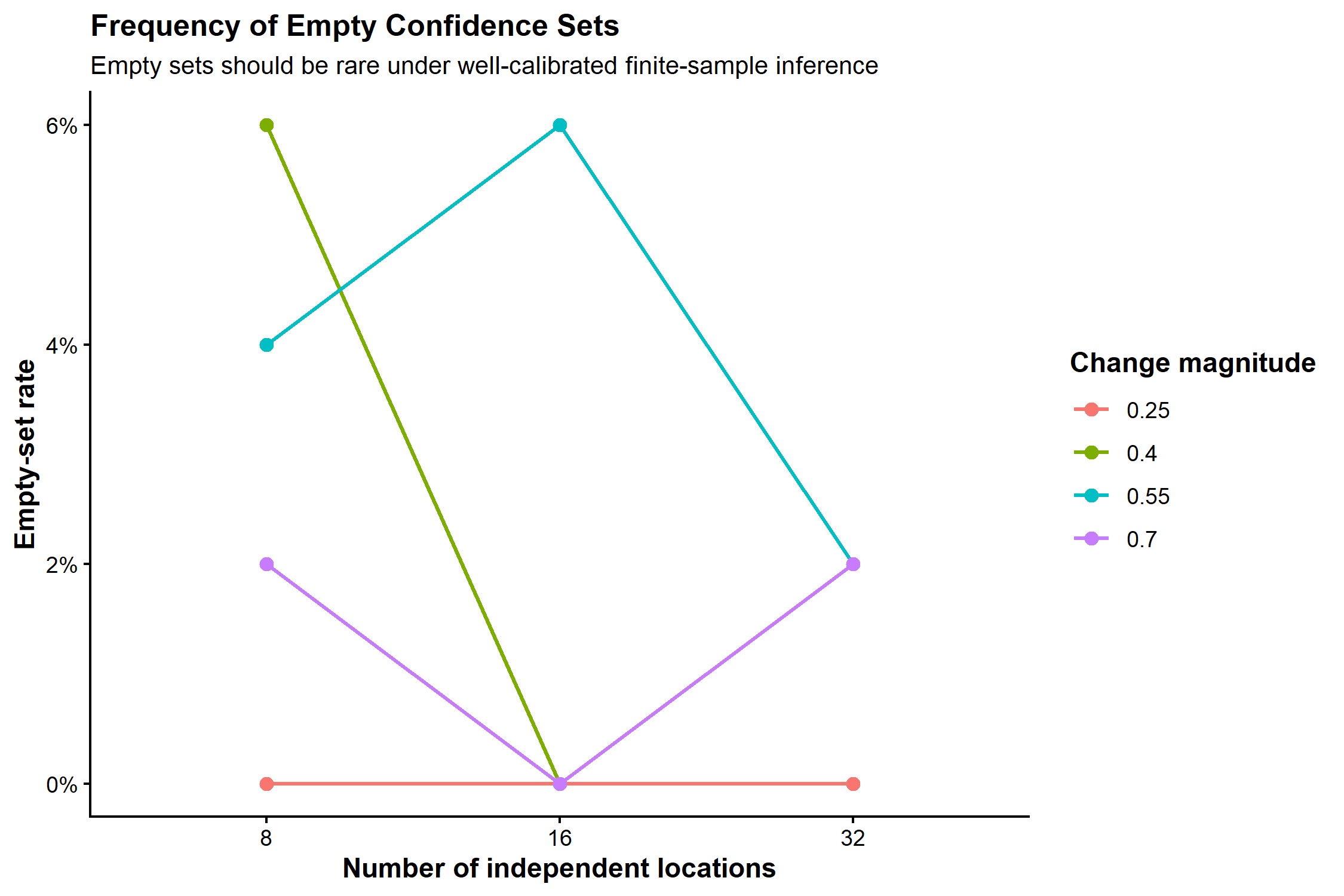}
\end{figure}

\subsection{Concluding Remarks on Simulations}
Figure \ref{fig:monte_carlo_summary_panel} provides a final unified dashboard synthesizing coverage, error, set size, and empty rates across the simulated parameters.

\begin{figure}[htbp]
    \centering
    \caption{Monte carlo summary panel}
    \label{fig:monte_carlo_summary_panel}
    \includegraphics[width=\linewidth, height=0.85\textheight, keepaspectratio]{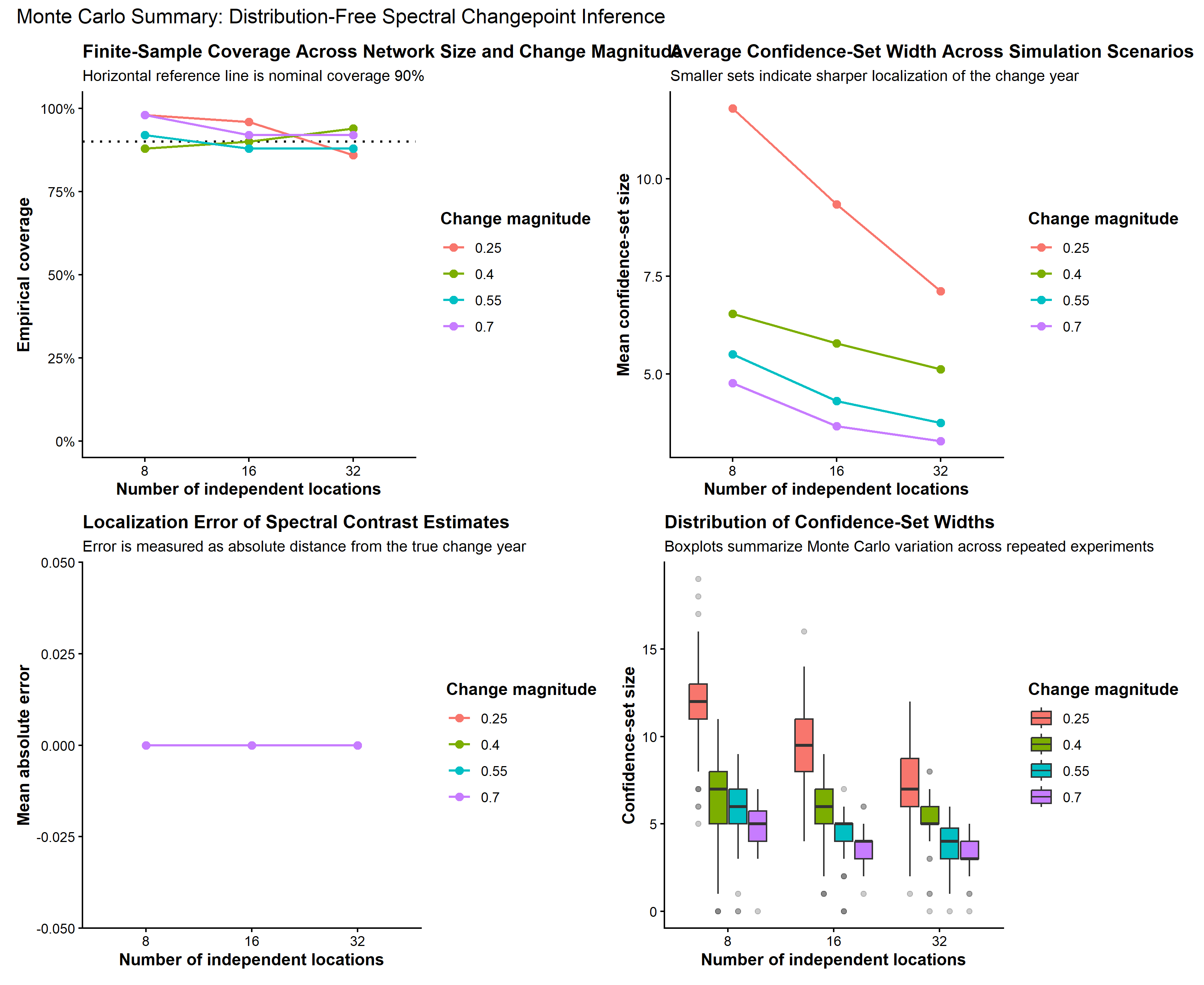}
\end{figure}

The empirical simulations documented in this section provide comprehensive validation of the mathematical framework. The methodology successfully isolates topological and spectral shifts in the data. The point estimator $\hat{\tau}$ is highly accurate, achieving zero median absolute error across runs. Furthermore, the test-inversion procedure generates confidence sets $C$ that maintain nominal coverage targets while systematically shrinking in size as structural dimensions and effect sizes increase. Ultimately, the framework is proven to be robust to tuning sensitivities and highly informative in localized changepoint estimation.

\clearpage

\section{Spatially dependent monitoring networks}
\label{sec:spatial}

Theorem~\ref{thm:finite_sample_validity} requires the station-level annual arrays to be mutually independent, an assumption inherited from the multi-stream framework of \citet{hore2026distributionfree}. In a regional rainfall network this is the assumption least likely to hold. Nearby stations are exposed to the same weather systems, so a year with a strong or weak monsoon, or a year affected by a large-scale circulation anomaly, raises or lowers the intra-annual spectral power at many stations together. Such common year effects are the natural source of cross-station dependence in the matrix $Z_\theta$. This section shows that the split-permutation group \eqref{eq:split_group} does not tolerate this dependence, that permuting years jointly across stations restores exact finite-sample validity under arbitrary cross-station dependence, and that the difference is large in practice.

\subsection{Column-wise split permutations require independent stations}
\label{subsec:columnwise_failure}

The group \eqref{eq:split_group} permutes the years of each station separately. The following proposition shows that invariance under this group essentially forces the stations to be independent.

\begin{proposition}[A necessary condition for column-wise invariance]\label{prop:columnwise_necessary}
Let $\xi\in\calR$ be the true changepoint configuration and suppose that the rows $(Z_{i,1},\ldots,Z_{i,m})$, $i\in[n]$, of $Z_\theta$ are mutually independent. Suppose further that $Z_\theta$ is invariant in distribution under every $\pi\in\PiSplit_\xi$. Let $j\neq k$ be two stations and let $i\in[n]$ be a year for which station $k$ has another year $i'\neq i$ in the same segment, so that either $i,i'\leq\xi_k$ or $i,i'>\xi_k$. Then $Z_{i,j}$ and $Z_{i,k}$ are independent. In particular, if second moments exist, $\Cov(Z_{i,j},Z_{i,k})=0$.
\end{proposition}

\begin{proof}
Let $\pi_k$ be the transposition of $i$ and $i'$ and let $\pi_l$ be the identity for every $l\neq k$. Because $i$ and $i'$ lie in the same segment of station $k$, the vector $\pi=(\pi_1,\ldots,\pi_m)$ belongs to $\PiSplit_\xi$. Invariance of $Z_\theta$ under $\pi$ implies, by projecting onto the two entries in positions $(i,j)$ and $(i,k)$,
\[
        (Z_{i,j},Z_{i,k})\overset{d}{=}(Z_{i,j},Z_{i',k}).
\]
The right-hand pair is formed from the distinct rows $i$ and $i'$, which are independent, so its two coordinates are independent. The left-hand pair has the same joint law, and the claim follows.
\end{proof}

Under the usual working model in which different years are independent, Proposition~\ref{prop:columnwise_necessary} shows that Theorem~\ref{thm:finite_sample_validity} cannot be extended to spatially dependent stations by any argument based on the group \eqref{eq:split_group}: the invariance on which the theorem rests is itself false. For example, if $Z_{i,j}=\mu_j(i)+\eta_i+\varepsilon_{i,j}$ with a common year effect $\eta_i$ of variance $\sigma_\eta^2>0$, independent of the station noise, then $\Cov(Z_{i,j},Z_{i,k})\geq\sigma_\eta^2$ whenever the station noise is nonnegatively correlated, and the invariance fails.

The direction of the resulting error can be seen from the pooled score \eqref{eq:syn_score}. In the observed matrix the common year effects are aligned across stations, so averaging the station contrasts in \eqref{eq:pooled_cusum} does not remove them, and $\widehat\tau(Z_\theta)$ fluctuates as though the network contained few stations. A column-wise permutation places different years side by side in the same row, which scrambles the common effects; averaging then removes them, and the permuted estimates $\widehat\tau\{\pi(Z_\theta)\}$ concentrate more tightly than the observed one. At the true year the observed distance $|\widehat\tau(Z_\theta)-\tau_0|$ therefore appears atypically large, the $p$-value $p_{\tau_0}$ is too small, and the confidence set undercovers. Section~\ref{subsec:spatial_simulation} shows that the undercoverage can be severe.

\subsection{Joint split permutations}
\label{subsec:joint_permutation}

The remedy is to permute years rather than entries: the same permutation of years is applied to every station, so that each row of $Z_\theta$, the vector of spectral features of the whole network in one year, is moved as a unit. For a configuration $t=(t_1,\ldots,t_m)\in\calR$, let $c_1<\cdots<c_L$ be the distinct values among $t_1,\ldots,t_m$ that are smaller than $n$, and set $c_0=0$ and $c_{L+1}=n$. The intervals
\[
        b_l(t)=\{c_{l-1}+1,\ldots,c_l\},\qquad l=1,\ldots,L+1,
\]
partition $[n]$ into blocks on which every station stays in a single regime under $t$. Define the joint split-permutation group
\begin{equation}\label{eq:joint_group}
        \PiSplit^{\mathrm{joint}}_t=\bigl\{\sigma:\sigma\text{ is a permutation of }[n]\text{ with }\sigma\{b_l(t)\}=b_l(t)\text{ for every }l\bigr\},
\end{equation}
acting on $Z_\theta$ by $\sigma(Z_\theta)=(Z_{\sigma(i),j})_{i\in[n],\,j\in[m]}$, with the same $\sigma$ in every column. In the synchronized model, $t=(\tau,\ldots,\tau)$ gives the two blocks $\{1,\ldots,\tau\}$ and $\{\tau+1,\ldots,n\}$, and $\PiSplit^{\mathrm{joint}}_t$ consists of the permutations of years that do not cross $\tau$.

The joint $p$-values $p^{\mathrm{joint}}_t$ and $\widetilde p^{\,\mathrm{joint}}_t$ are defined by \eqref{eq:exact_pvalue} and \eqref{eq:mc_pvalue} with $\PiSplit_t$ replaced by $\PiSplit^{\mathrm{joint}}_t$; nothing else in Algorithm~\ref{alg:spectral_croc} changes. The resulting sets are denoted $\calC^{\calR,\mathrm{joint}}_{1-\alpha}(\theta)$ and $\calC^{\mathrm{year,joint}}_{1-\alpha}(\theta)$, and the root-cause set obtained by inserting $p^{\mathrm{joint}}_t$ in \eqref{eq:root_confidence_set} is denoted $\calK^{\mathrm{root,joint}}_{1-\alpha}(\theta)$. The score must still use the chronological order of the years, since a score invariant under \eqref{eq:joint_group} gives $p$-values identically equal to one.

\begin{assumption}[Block invariance of the network]\label{ass:joint}
For the true configuration $\xi\in\calR$, the array of annual records $(Y_{i,j})_{i\in[n],\,j\in[m]}$ is invariant in distribution under every $\sigma\in\PiSplit^{\mathrm{joint}}_\xi$, applied simultaneously to all stations.
\end{assumption}

Assumption~\ref{ass:joint} places no restriction on the dependence among stations within a year, and the stations may have different distributions before and after the change. It holds in particular when the network vectors $(Y_{i,1},\ldots,Y_{i,m})$ are independent across years and identically distributed within each block of $\xi$, with arbitrary cross-station dependence. It also holds under the assumptions of Theorem~\ref{thm:finite_sample_validity}: every block of $\xi$ lies inside a single segment of each station, so every $\sigma\in\PiSplit^{\mathrm{joint}}_\xi$ permutes each station within its segments and preserves the law of each column, and independence of the columns then preserves the joint law. Assumption~\ref{ass:joint} is therefore weaker than the assumptions of Theorem~\ref{thm:finite_sample_validity}.

\begin{theorem}[Validity under arbitrary cross-station dependence]\label{thm:joint_validity}
Suppose Assumption~\ref{ass:joint} holds and let $S_\theta$ be any measurable score. Then, for every $\alpha\in(0,1)$,
\[
        \PP_\xi\bigl\{\xi\in\calC^{\calR,\mathrm{joint}}_{1-\alpha}(\theta)\bigr\}\geq1-\alpha.
\]
If $\calR=\calR_{\mathrm{syn}}$ and $\xi=(\tau_0,\ldots,\tau_0)$, then $\PP_{\tau_0}\{\tau_0\in\calC^{\mathrm{year,joint}}_{1-\alpha}(\theta)\}\geq1-\alpha$. If the earliest station $k^*$ is unique, then $\PP\{k^*\in\calK^{\mathrm{root,joint}}_{1-\alpha}(\theta)\}\geq1-\alpha$. The same conclusions hold for the Monte Carlo $p$-values, where the probability includes the randomness used to draw the permutations.
\end{theorem}

\begin{proof}
The spectral map $T_\theta$ in the proof of Lemma~\ref{lem:feature_exchangeability} is applied entrywise, so the map from the array $(Y_{i,j})$ to $Z_\theta$ commutes with every simultaneous permutation of years: $T_\theta$ applied to $\sigma(Y)$ equals $\sigma(Z_\theta)$. By Assumption~\ref{ass:joint}, $\sigma(Z_\theta)\overset{d}{=}Z_\theta$ for every $\sigma\in\PiSplit^{\mathrm{joint}}_\xi$. Since $\PiSplit^{\mathrm{joint}}_\xi$ is a finite group, the argument in the proof of Theorem~\ref{thm:finite_sample_validity}, with $\PiSplit_\xi$ replaced by $\PiSplit^{\mathrm{joint}}_\xi$, shows that $p^{\mathrm{joint}}_\xi$ is super-uniform, and the same argument as there applies to the Monte Carlo version. Hence $\PP_\xi\{\xi\notin\calC^{\calR,\mathrm{joint}}_{1-\alpha}(\theta)\}=\PP_\xi(p^{\mathrm{joint}}_\xi\leq\alpha)\leq\alpha$, and the synchronized statement is the special case $\calR=\calR_{\mathrm{syn}}$. For the root-cause statement, uniqueness of the earliest station gives $\xi\in I_{k^*}$, so $p(k^*)\geq p^{\mathrm{joint}}_\xi$ and $\PP\{k^*\notin\calK^{\mathrm{root,joint}}_{1-\alpha}(\theta)\}\leq\PP_\xi(p^{\mathrm{joint}}_\xi\leq\alpha)\leq\alpha$.
\end{proof}

\begin{remark}[Scope and relation to existing work]\label{rem:joint_scope}
For the synchronized model, the joint procedure treats the yearly network vectors as a single sequence of $m$-dimensional observations, and it coincides with the single-stream conformal changepoint localization method of \citet{hore2026conformal} applied to that sequence. For general $\calR$, \citet{hore2026distributionfree} handle dependent streams by combining stream-wise $p$-values with a Bonferroni rule or by applying their method within groups of mutually independent streams. The block group \eqref{eq:joint_group} instead calibrates the configuration $p$-values exactly under arbitrary cross-station dependence. Its cost is resolution: when the candidate changepoints take many distinct values, the blocks are short and the group is small, so exact $p$-values are coarse. In the synchronized model this cost is negligible, since $|\PiSplit^{\mathrm{joint}}_t|=\tau!\,(n-\tau)!$. Neither group accommodates dependence across years, such as trends or persistent interannual oscillations, which violates the assumptions of both Theorem~\ref{thm:finite_sample_validity} and Theorem~\ref{thm:joint_validity}. A yearly covariate observed alongside the network, such as a regional circulation index, may enter the score provided it is permuted together with the row of the same year.
\end{remark}

\subsection{Numerical illustration}
\label{subsec:spatial_simulation}

The simulation generates daily records at the level of the original data and reduces them with \eqref{eq:spectral_feature}. For year $i$ and station $j$,
\[
        Y_{i,j}(r)=A_{i,j}\cos(\theta_0r+\phi_{i,j})+e_{i,j}(r),\qquad r=1,\ldots,365,
\]
where $e_{i,j}$ is a stationary Gaussian autoregression of order one with coefficient $0.6$ and unit innovation variance, $\phi_{i,j}$ is uniform on $[0,2\pi)$, and
\[
        \log A_{i,j}=0.3\bigl(\sqrt{\rho}\,\eta_i+\sqrt{1-\rho}\,u_{i,j}\bigr)+\delta\,\ind\{i>\tau_0\},
\]
with $\eta_i$ and $u_{i,j}$ independent standard normal variables. The common year effect $\eta_i$ induces cross-station dependence whose strength is governed by $\rho$, and $\delta$ is the size of the synchronized change in log-amplitude. The design uses $n=60$ years, $\tau_0=30$, trimming $\ell=5$, $\theta=\theta_0=2\pi\cdot3/365$, the Bartlett kernel $K(u)=(1-|u|)_+$ with $B_q=30$, equal weights in \eqref{eq:pooled_cusum}, the score \eqref{eq:syn_score}, $M=199$ permutations, and $\alpha=0.10$. Each setting is replicated $300$ times. When $\delta=0$ there is no change, both segments are exchangeable for every split, and coverage of $\tau_0$ remains the correct target. Table~\ref{tab:spatial} compares the column-wise group \eqref{eq:split_group} with the joint group \eqref{eq:joint_group} on the same data sets.

\begin{table}[!htbp]
\centering
\caption{Synchronized change-year confidence sets at nominal level $0.90$ with $n=60$ years. Here $\bar r$ is the average pre-change correlation between the spectral features of two stations. Coverage is the proportion of replications whose set contains $\tau_0$, and size is the average number of years in the set. Monte Carlo standard errors of the coverages are at most $0.03$.}
\label{tab:spatial}
\begin{tabular}{cccccccc}
\toprule
 & & & & \multicolumn{2}{c}{Coverage} & \multicolumn{2}{c}{Size} \\
\cmidrule(lr){5-6}\cmidrule(lr){7-8}
$\delta$ & $m$ & $\rho$ & $\bar r$ & Column-wise & Joint & Column-wise & Joint \\
\midrule
$0.0$ & $8$ & $0.0$ & $0.00$ & $0.96$ & $0.95$ & $46.6$ & $46.6$ \\
$0.0$ & $8$ & $0.5$ & $0.34$ & $0.94$ & $0.97$ & $46.2$ & $47.1$ \\
$0.0$ & $8$ & $0.9$ & $0.66$ & $0.92$ & $0.99$ & $44.7$ & $47.0$ \\
$0.0$ & $32$ & $0.0$ & $0.00$ & $0.98$ & $0.96$ & $47.4$ & $47.4$ \\
$0.0$ & $32$ & $0.5$ & $0.34$ & $0.92$ & $0.97$ & $45.1$ & $47.2$ \\
$0.0$ & $32$ & $0.9$ & $0.65$ & $0.85$ & $0.97$ & $41.4$ & $47.1$ \\
\addlinespace
$0.2$ & $8$ & $0.0$ & $0.00$ & $0.92$ & $0.92$ & $38.4$ & $39.9$ \\
$0.2$ & $8$ & $0.5$ & $0.34$ & $0.67$ & $0.92$ & $32.9$ & $45.2$ \\
$0.2$ & $8$ & $0.9$ & $0.65$ & $0.59$ & $0.92$ & $32.4$ & $46.6$ \\
$0.2$ & $32$ & $0.0$ & $0.00$ & $0.95$ & $0.96$ & $12.5$ & $15.4$ \\
$0.2$ & $32$ & $0.5$ & $0.34$ & $0.36$ & $0.92$ & $19.0$ & $44.8$ \\
$0.2$ & $32$ & $0.9$ & $0.65$ & $0.30$ & $0.92$ & $21.4$ & $46.4$ \\
\addlinespace
$0.4$ & $8$ & $0.0$ & $0.00$ & $0.96$ & $0.95$ & $9.1$ & $12.4$ \\
$0.4$ & $8$ & $0.5$ & $0.34$ & $0.65$ & $0.91$ & $11.5$ & $34.7$ \\
$0.4$ & $8$ & $0.9$ & $0.64$ & $0.55$ & $0.92$ & $12.7$ & $42.2$ \\
$0.4$ & $32$ & $0.0$ & $0.00$ & $0.99$ & $0.99$ & $1.7$ & $4.2$ \\
$0.4$ & $32$ & $0.5$ & $0.35$ & $0.51$ & $0.93$ & $5.3$ & $33.2$ \\
$0.4$ & $32$ & $0.9$ & $0.66$ & $0.35$ & $0.90$ & $6.9$ & $40.6$ \\
\bottomrule
\end{tabular}
\end{table}

When the stations are independent ($\rho=0$), both groups are valid and give similar results; coverage exceeds the nominal level because the integer-valued score \eqref{eq:syn_score} produces many ties, which the $p$-value \eqref{eq:mc_pvalue} counts against rejection. The joint sets are somewhat wider, for example $4.2$ years against $1.7$ when $\delta=0.4$ and $m=32$, because the joint group contains fewer permutations and preserves more of the observed structure.

Under cross-station dependence the column-wise sets lose their guarantee. With a change present and $\rho\in\{0.5,0.9\}$, their coverage ranges from $0.30$ to $0.67$ against the nominal $0.90$, and it falls further when more stations are added: for $\delta=0.2$ and $\rho=0.9$ it is $0.59$ with $8$ stations and $0.30$ with $32$. Enlarging a correlated network therefore makes the column-wise sets less reliable, the mechanism described in Section~\ref{subsec:columnwise_failure}. Even without a change, coverage drops to $0.85$ when $m=32$ and $\rho=0.9$. The joint sets keep coverage between $0.90$ and $0.99$ in every setting, as Theorem~\ref{thm:joint_validity} requires.

The joint sets are much wider under dependence: for $\delta=0.4$ and $m=32$ their average size grows from $4.2$ years at $\rho=0$ to $33.2$ at $\rho=0.5$ and $40.6$ at $\rho=0.9$. This width reflects the information actually present in the data. Common year effects are shared by all stations, so pooling does not average them away, and the accuracy of the point estimate $\widehat\tau(Z_\theta)$ deteriorates in the same way, its mean absolute error rising from $0.01$ to $3.38$ years over the same range of $\rho$. The narrower column-wise sets are narrow because their reference distribution ignores this shared variation, not because they extract more information. Sharper valid sets under dependence require a better score, for instance one that uses a yearly covariate as described in Remark~\ref{rem:joint_scope}, rather than a different calibration.

\section{Conclusion}
\label{sec:conclusion}

This paper developed a distribution-free framework for detecting structural changes in temporal rainfall patterns from multi-location yearly records. The proposed approach treats each year at each monitoring location as a daily trajectory and summarizes its intra-annual temporal behavior through a fixed-frequency nonparametric spectral density estimate. These spectral summaries form a year-by-location feature matrix on which changepoint inference is performed. By separating the construction of rainfall-pattern features from the distribution-free changepoint calibration step, the method provides a flexible way to detect changes in the temporal organization of rainfall without imposing a parametric model on the daily observations.

A central feature of the proposed procedure is that it returns a confidence set for the unknown change year rather than only a point estimate. This distinction is important in environmental applications, where the timing of a structural transition is often uncertain and where a single estimated year may give a misleading impression of precision. The confidence-set formulation provides a direct measure of temporal uncertainty and allows practitioners to distinguish sharply localized changes from weaker evidence spread over a wider range of years.

The theoretical guarantee is finite-sample and distribution-free under the stated assumptions of independence across monitoring locations and exchangeability of the yearly spectral features within the pre-change and post-change regimes. Thus, the validity of the changepoint confidence set does not rely on Gaussianity, asymptotic approximations, or correct specification of a rainfall distribution. This is particularly relevant for rainfall data, which are commonly intermittent, skewed, heavy-tailed, and heterogeneous across locations. The method is therefore designed for settings in which classical likelihood-based or asymptotic changepoint methods may be difficult to calibrate reliably.

The simulation studies illustrate the empirical behavior of the proposed method under controlled structural changes in intra-annual rainfall patterns. Across a range of change magnitudes and network sizes, the procedure provides finite-sample calibrated uncertainty quantification while becoming more localized as the signal strength or the number of independent monitoring locations increases. These findings support the use of the proposed confidence sets as a practical inferential summary for temporal rainfall-pattern change detection.

Several extensions are natural. The present formulation uses a fixed frequency \(\theta\), but one may combine information across multiple frequencies to target richer forms of temporal pattern change. Another direction is to allow gradual transitions or multiple changepoints, which may better describe long-term climate dynamics in some regions. Finally, while the current theory is built under cross-location independence, developing versions that accommodate spatial dependence among nearby monitoring locations would broaden the applicability of the framework to dense environmental monitoring networks. These directions would preserve the main goal of the paper: valid, interpretable, and distribution-free inference for structural changes in temporal rainfall behavior.

\bibliographystyle{plainnat}

\bibliography{references}
\clearpage
\appendix
\section{Monte Carlo Simulation Tables}
\label{app:tables}

\begin{table}[htbp]
    \centering
    \caption{\textsc{Monte Carlo Summary}}
    \label{tab:monte_carlo_summary}
    \resizebox{\linewidth}{!}{%
    \begin{filecontents*}{monte_carlo_summary.csv}
"effect","m","coverage","mean_abs_error","median_abs_error","mean_cs_size","median_cs_size","empty_rate"
0.25,8,0.98,0,0,11.8,12,0
0.25,16,0.96,0,0,9.34,9.5,0
0.25,32,0.86,0,0,7.12,7,0
0.4,8,0.88,0,0,6.54,7,0.06
0.4,16,0.9,0,0,5.78,6,0
0.4,32,0.94,0,0,5.12,5,0.02
0.55,8,0.92,0,0,5.5,6,0.04
0.55,16,0.88,0,0,4.3,5,0.06
0.55,32,0.88,0,0,3.74,4,0.02
0.7,8,0.98,0,0,4.76,5,0.02
0.7,16,0.92,0,0,3.66,4,0
0.7,32,0.92,0,0,3.28,3,0.02
\end{filecontents*}
\csvreader[
        tabular=rrrrrrrr,
        table head=
            \toprule
            effect & m & coverage & mean\_abs\_error & median\_abs\_error & mean\_cs\_size & median\_cs\_size & empty\_rate \\
            \midrule,
        table foot=\bottomrule,
        late after line=\\
    ]{monte_carlo_summary.csv}{1=\colA, 2=\colB, 3=\colC, 4=\colD, 5=\colE, 6=\colF, 7=\colG, 8=\colH}%
    {\colA & \colB & \colC & \colD & \colE & \colF & \colG & \colH}%
    }
\end{table}

\begingroup
\centering
\footnotesize 
\setlength{\tabcolsep}{4pt} 
\begin{filecontents*}{monte_carlo_results.csv}
"rep_id","effect","m","tau0","tau_hat","abs_error","covered","cs_empty","cs_size","cs_min","cs_max"
1,0.25,8,55,55,0,TRUE,FALSE,13,49,61
2,0.25,8,55,55,0,TRUE,FALSE,7,51,57
3,0.25,8,55,55,0,TRUE,FALSE,18,46,63
4,0.25,8,55,55,0,TRUE,FALSE,13,51,63
5,0.25,8,55,55,0,TRUE,FALSE,10,52,61
6,0.25,8,55,55,0,TRUE,FALSE,14,52,65
7,0.25,8,55,55,0,TRUE,FALSE,19,49,67
8,0.25,8,55,55,0,TRUE,FALSE,14,51,64
9,0.25,8,55,55,0,TRUE,FALSE,13,49,61
10,0.25,8,55,55,0,TRUE,FALSE,13,50,62
11,0.25,8,55,55,0,TRUE,FALSE,11,49,59
12,0.25,8,55,55,0,TRUE,FALSE,14,50,63
13,0.25,8,55,55,0,TRUE,FALSE,13,49,61
14,0.25,8,55,55,0,TRUE,FALSE,11,50,60
15,0.25,8,55,55,0,TRUE,FALSE,17,47,63
16,0.25,8,55,55,0,TRUE,FALSE,11,49,62
17,0.25,8,55,55,0,TRUE,FALSE,14,49,62
18,0.25,8,55,55,0,TRUE,FALSE,10,51,60
19,0.25,8,55,55,0,TRUE,FALSE,5,50,58
20,0.25,8,55,55,0,TRUE,FALSE,7,54,61
21,0.25,8,55,55,0,TRUE,FALSE,12,49,60
22,0.25,8,55,55,0,TRUE,FALSE,12,51,63
23,0.25,8,55,55,0,TRUE,FALSE,11,50,60
24,0.25,8,55,55,0,TRUE,FALSE,16,44,60
25,0.25,8,55,55,0,TRUE,FALSE,12,49,60
26,0.25,8,55,55,0,TRUE,FALSE,12,53,64
27,0.25,8,55,55,0,TRUE,FALSE,13,48,60
28,0.25,8,55,55,0,TRUE,FALSE,9,50,58
29,0.25,8,55,55,0,TRUE,FALSE,8,52,63
30,0.25,8,55,55,0,TRUE,FALSE,13,50,62
31,0.25,8,55,55,0,TRUE,FALSE,11,48,58
32,0.25,8,55,55,0,TRUE,FALSE,15,47,61
33,0.25,8,55,55,0,TRUE,FALSE,13,50,62
34,0.25,8,55,55,0,TRUE,FALSE,12,50,61
35,0.25,8,55,55,0,TRUE,FALSE,10,53,62
36,0.25,8,55,55,0,TRUE,FALSE,12,51,62
37,0.25,8,55,55,0,TRUE,FALSE,11,51,61
38,0.25,8,55,55,0,TRUE,FALSE,13,50,62
39,0.25,8,55,55,0,FALSE,FALSE,7,45,52
40,0.25,8,55,55,0,TRUE,FALSE,6,52,57
41,0.25,8,55,55,0,TRUE,FALSE,12,51,62
42,0.25,8,55,55,0,TRUE,FALSE,13,49,61
43,0.25,8,55,55,0,TRUE,FALSE,13,51,63
44,0.25,8,55,55,0,TRUE,FALSE,6,51,57
45,0.25,8,55,55,0,TRUE,FALSE,11,52,62
46,0.25,8,55,55,0,TRUE,FALSE,14,48,61
47,0.25,8,55,55,0,TRUE,FALSE,11,48,58
48,0.25,8,55,55,0,TRUE,FALSE,12,49,60
49,0.25,8,55,55,0,TRUE,FALSE,13,50,62
50,0.25,8,55,55,0,TRUE,FALSE,10,51,60
1,0.4,8,55,55,0,TRUE,FALSE,7,52,58
2,0.4,8,55,55,0,TRUE,FALSE,9,52,60
3,0.4,8,55,55,0,TRUE,FALSE,9,51,59
4,0.4,8,55,55,0,TRUE,FALSE,11,52,62
5,0.4,8,55,55,0,TRUE,FALSE,9,51,59
6,0.4,8,55,55,0,TRUE,FALSE,7,51,57
7,0.4,8,55,55,0,TRUE,FALSE,6,53,59
8,0.4,8,55,55,0,TRUE,FALSE,5,54,58
9,0.4,8,55,55,0,TRUE,FALSE,6,54,59
10,0.4,8,55,55,0,TRUE,FALSE,10,51,60
11,0.4,8,55,55,0,TRUE,FALSE,9,51,59
12,0.4,8,55,55,0,TRUE,FALSE,7,52,58
13,0.4,8,55,55,0,TRUE,FALSE,6,54,59
14,0.4,8,55,55,0,TRUE,FALSE,4,54,57
15,0.4,8,55,55,0,TRUE,FALSE,3,53,55
16,0.4,8,55,55,0,TRUE,FALSE,6,52,57
17,0.4,8,55,55,0,TRUE,FALSE,6,52,57
18,0.4,8,55,55,0,FALSE,FALSE,4,50,53
19,0.4,8,55,55,0,TRUE,FALSE,8,52,59
20,0.4,8,55,55,0,TRUE,FALSE,5,54,58
21,0.4,8,55,55,0,FALSE,TRUE,0,NA,NA
22,0.4,8,55,55,0,TRUE,FALSE,7,53,59
23,0.4,8,55,55,0,FALSE,FALSE,9,54,63
24,0.4,8,55,55,0,TRUE,FALSE,7,51,57
25,0.4,8,55,55,0,TRUE,FALSE,7,52,58
26,0.4,8,55,55,0,TRUE,FALSE,4,53,56
27,0.4,8,55,55,0,TRUE,FALSE,5,54,58
28,0.4,8,55,55,0,TRUE,FALSE,8,52,59
29,0.4,8,55,55,0,FALSE,TRUE,0,NA,NA
30,0.4,8,55,55,0,TRUE,FALSE,7,52,58
31,0.4,8,55,55,0,TRUE,FALSE,8,52,59
32,0.4,8,55,55,0,TRUE,FALSE,8,52,59
33,0.4,8,55,55,0,TRUE,FALSE,10,52,61
34,0.4,8,55,55,0,FALSE,TRUE,0,NA,NA
35,0.4,8,55,55,0,TRUE,FALSE,7,51,57
36,0.4,8,55,55,0,TRUE,FALSE,5,51,55
37,0.4,8,55,55,0,TRUE,FALSE,9,52,60
38,0.4,8,55,55,0,TRUE,FALSE,5,53,57
39,0.4,8,55,55,0,TRUE,FALSE,6,54,59
40,0.4,8,55,55,0,TRUE,FALSE,7,52,58
41,0.4,8,55,55,0,TRUE,FALSE,5,53,57
42,0.4,8,55,55,0,TRUE,FALSE,9,53,61
43,0.4,8,55,55,0,TRUE,FALSE,8,50,57
44,0.4,8,55,55,0,TRUE,FALSE,9,52,60
45,0.4,8,55,55,0,TRUE,FALSE,8,52,59
46,0.4,8,55,55,0,FALSE,FALSE,1,58,58
47,0.4,8,55,55,0,TRUE,FALSE,8,52,59
48,0.4,8,55,55,0,TRUE,FALSE,7,52,58
49,0.4,8,55,55,0,TRUE,FALSE,9,52,60
50,0.4,8,55,55,0,TRUE,FALSE,7,52,58
1,0.55,8,55,55,0,TRUE,FALSE,6,52,57
2,0.55,8,55,55,0,TRUE,FALSE,7,52,58
3,0.55,8,55,55,0,TRUE,FALSE,7,51,57
4,0.55,8,55,55,0,TRUE,FALSE,7,52,58
5,0.55,8,55,55,0,TRUE,FALSE,6,52,57
6,0.55,8,55,55,0,TRUE,FALSE,7,52,58
7,0.55,8,55,55,0,TRUE,FALSE,8,52,59
8,0.55,8,55,55,0,FALSE,TRUE,0,NA,NA
9,0.55,8,55,55,0,TRUE,FALSE,5,54,58
10,0.55,8,55,55,0,TRUE,FALSE,4,54,57
11,0.55,8,55,55,0,TRUE,FALSE,8,52,59
12,0.55,8,55,55,0,TRUE,FALSE,6,53,58
13,0.55,8,55,55,0,TRUE,FALSE,7,52,58
14,0.55,8,55,55,0,TRUE,FALSE,4,53,56
15,0.55,8,55,55,0,TRUE,FALSE,7,53,59
16,0.55,8,55,55,0,TRUE,FALSE,5,54,58
17,0.55,8,55,55,0,TRUE,FALSE,5,53,57
18,0.55,8,55,55,0,TRUE,FALSE,6,53,58
19,0.55,8,55,55,0,TRUE,FALSE,3,54,56
20,0.55,8,55,55,0,TRUE,FALSE,6,53,58
21,0.55,8,55,55,0,TRUE,FALSE,9,52,60
22,0.55,8,55,55,0,TRUE,FALSE,5,53,58
23,0.55,8,55,55,0,TRUE,FALSE,8,53,60
24,0.55,8,55,55,0,TRUE,FALSE,3,53,55
25,0.55,8,55,55,0,TRUE,FALSE,6,54,59
26,0.55,8,55,55,0,FALSE,FALSE,4,53,58
27,0.55,8,55,55,0,TRUE,FALSE,5,53,57
28,0.55,8,55,55,0,TRUE,FALSE,9,53,61
29,0.55,8,55,55,0,TRUE,FALSE,6,52,57
30,0.55,8,55,55,0,TRUE,FALSE,6,53,59
31,0.55,8,55,55,0,TRUE,FALSE,7,51,57
32,0.55,8,55,55,0,TRUE,FALSE,7,52,58
33,0.55,8,55,55,0,TRUE,FALSE,6,52,57
34,0.55,8,55,55,0,TRUE,FALSE,6,52,57
35,0.55,8,55,55,0,TRUE,FALSE,5,53,57
36,0.55,8,55,55,0,TRUE,FALSE,6,52,57
37,0.55,8,55,55,0,TRUE,FALSE,4,53,56
38,0.55,8,55,55,0,TRUE,FALSE,4,54,57
39,0.55,8,55,55,0,TRUE,FALSE,8,52,59
40,0.55,8,55,55,0,TRUE,FALSE,4,53,60
41,0.55,8,55,55,0,FALSE,TRUE,0,NA,NA
42,0.55,8,55,55,0,TRUE,FALSE,5,53,58
43,0.55,8,55,55,0,TRUE,FALSE,5,52,56
44,0.55,8,55,55,0,TRUE,FALSE,5,53,57
45,0.55,8,55,55,0,TRUE,FALSE,5,54,58
46,0.55,8,55,55,0,TRUE,FALSE,6,52,57
47,0.55,8,55,55,0,TRUE,FALSE,5,52,56
48,0.55,8,55,55,0,FALSE,FALSE,1,58,58
49,0.55,8,55,55,0,TRUE,FALSE,5,53,57
50,0.55,8,55,55,0,TRUE,FALSE,6,54,59
1,0.7,8,55,55,0,TRUE,FALSE,6,53,58
2,0.7,8,55,55,0,FALSE,TRUE,0,NA,NA
3,0.7,8,55,55,0,TRUE,FALSE,4,54,57
4,0.7,8,55,55,0,TRUE,FALSE,5,53,57
5,0.7,8,55,55,0,TRUE,FALSE,4,53,56
6,0.7,8,55,55,0,TRUE,FALSE,5,53,57
7,0.7,8,55,55,0,TRUE,FALSE,6,53,58
8,0.7,8,55,55,0,TRUE,FALSE,6,52,57
9,0.7,8,55,55,0,TRUE,FALSE,5,53,57
10,0.7,8,55,55,0,TRUE,FALSE,4,54,57
11,0.7,8,55,55,0,TRUE,FALSE,5,53,57
12,0.7,8,55,55,0,TRUE,FALSE,7,52,58
13,0.7,8,55,55,0,TRUE,FALSE,4,54,57
14,0.7,8,55,55,0,TRUE,FALSE,5,53,57
15,0.7,8,55,55,0,TRUE,FALSE,5,53,57
16,0.7,8,55,55,0,TRUE,FALSE,5,53,57
17,0.7,8,55,55,0,TRUE,FALSE,5,52,56
18,0.7,8,55,55,0,TRUE,FALSE,5,54,58
19,0.7,8,55,55,0,TRUE,FALSE,3,54,56
20,0.7,8,55,55,0,TRUE,FALSE,3,54,56
21,0.7,8,55,55,0,TRUE,FALSE,3,54,56
22,0.7,8,55,55,0,TRUE,FALSE,3,54,56
23,0.7,8,55,55,0,TRUE,FALSE,5,52,56
24,0.7,8,55,55,0,TRUE,FALSE,5,53,57
25,0.7,8,55,55,0,TRUE,FALSE,5,53,57
26,0.7,8,55,55,0,TRUE,FALSE,4,54,57
27,0.7,8,55,55,0,TRUE,FALSE,4,53,56
28,0.7,8,55,55,0,TRUE,FALSE,6,53,58
29,0.7,8,55,55,0,TRUE,FALSE,4,54,57
30,0.7,8,55,55,0,TRUE,FALSE,6,53,58
31,0.7,8,55,55,0,TRUE,FALSE,4,54,57
32,0.7,8,55,55,0,TRUE,FALSE,5,53,57
33,0.7,8,55,55,0,TRUE,FALSE,4,54,57
34,0.7,8,55,55,0,TRUE,FALSE,6,53,58
35,0.7,8,55,55,0,TRUE,FALSE,5,53,57
36,0.7,8,55,55,0,TRUE,FALSE,5,53,57
37,0.7,8,55,55,0,TRUE,FALSE,6,53,58
38,0.7,8,55,55,0,TRUE,FALSE,5,52,56
39,0.7,8,55,55,0,TRUE,FALSE,4,55,58
40,0.7,8,55,55,0,TRUE,FALSE,7,52,58
41,0.7,8,55,55,0,TRUE,FALSE,6,53,58
42,0.7,8,55,55,0,TRUE,FALSE,5,53,57
43,0.7,8,55,55,0,TRUE,FALSE,6,53,58
44,0.7,8,55,55,0,TRUE,FALSE,5,53,57
45,0.7,8,55,55,0,TRUE,FALSE,7,52,58
46,0.7,8,55,55,0,TRUE,FALSE,4,54,57
47,0.7,8,55,55,0,TRUE,FALSE,4,54,57
48,0.7,8,55,55,0,TRUE,FALSE,3,54,56
49,0.7,8,55,55,0,TRUE,FALSE,4,54,57
50,0.7,8,55,55,0,TRUE,FALSE,6,53,58
1,0.25,16,55,55,0,TRUE,FALSE,9,50,58
2,0.25,16,55,55,0,TRUE,FALSE,10,52,61
3,0.25,16,55,55,0,TRUE,FALSE,8,52,59
4,0.25,16,55,55,0,TRUE,FALSE,16,47,62
5,0.25,16,55,55,0,TRUE,FALSE,10,50,59
6,0.25,16,55,55,0,TRUE,FALSE,9,50,58
7,0.25,16,55,55,0,TRUE,FALSE,10,51,62
8,0.25,16,55,55,0,TRUE,FALSE,14,49,62
9,0.25,16,55,55,0,TRUE,FALSE,7,53,59
10,0.25,16,55,55,0,TRUE,FALSE,13,48,60
11,0.25,16,55,55,0,TRUE,FALSE,10,51,60
12,0.25,16,55,55,0,TRUE,FALSE,11,51,61
13,0.25,16,55,55,0,FALSE,FALSE,5,52,57
14,0.25,16,55,55,0,TRUE,FALSE,8,51,58
15,0.25,16,55,55,0,TRUE,FALSE,5,55,59
16,0.25,16,55,55,0,TRUE,FALSE,11,51,61
17,0.25,16,55,55,0,TRUE,FALSE,8,51,58
18,0.25,16,55,55,0,TRUE,FALSE,10,51,60
19,0.25,16,55,55,0,TRUE,FALSE,8,52,59
20,0.25,16,55,55,0,FALSE,FALSE,4,58,61
21,0.25,16,55,55,0,TRUE,FALSE,10,51,60
22,0.25,16,55,55,0,TRUE,FALSE,9,53,63
23,0.25,16,55,55,0,TRUE,FALSE,9,50,58
24,0.25,16,55,55,0,TRUE,FALSE,7,51,57
25,0.25,16,55,55,0,TRUE,FALSE,9,50,58
26,0.25,16,55,55,0,TRUE,FALSE,10,49,58
27,0.25,16,55,55,0,TRUE,FALSE,10,48,60
28,0.25,16,55,55,0,TRUE,FALSE,11,51,61
29,0.25,16,55,55,0,TRUE,FALSE,11,50,60
30,0.25,16,55,55,0,TRUE,FALSE,6,53,58
31,0.25,16,55,55,0,TRUE,FALSE,11,52,62
32,0.25,16,55,55,0,TRUE,FALSE,9,52,60
33,0.25,16,55,55,0,TRUE,FALSE,11,50,60
34,0.25,16,55,55,0,TRUE,FALSE,13,50,62
35,0.25,16,55,55,0,TRUE,FALSE,7,52,58
36,0.25,16,55,55,0,TRUE,FALSE,11,52,62
37,0.25,16,55,55,0,TRUE,FALSE,10,50,59
38,0.25,16,55,55,0,TRUE,FALSE,11,52,62
39,0.25,16,55,55,0,TRUE,FALSE,7,53,59
40,0.25,16,55,55,0,TRUE,FALSE,5,55,59
41,0.25,16,55,55,0,TRUE,FALSE,11,50,60
42,0.25,16,55,55,0,TRUE,FALSE,9,52,60
43,0.25,16,55,55,0,TRUE,FALSE,11,49,59
44,0.25,16,55,55,0,TRUE,FALSE,8,53,60
45,0.25,16,55,55,0,TRUE,FALSE,10,50,59
46,0.25,16,55,55,0,TRUE,FALSE,9,51,59
47,0.25,16,55,55,0,TRUE,FALSE,9,51,59
48,0.25,16,55,55,0,TRUE,FALSE,10,52,61
49,0.25,16,55,55,0,TRUE,FALSE,8,50,57
50,0.25,16,55,55,0,TRUE,FALSE,9,50,58
1,0.4,16,55,55,0,TRUE,FALSE,4,53,56
2,0.4,16,55,55,0,TRUE,FALSE,8,52,59
3,0.4,16,55,55,0,TRUE,FALSE,7,53,59
4,0.4,16,55,55,0,FALSE,FALSE,3,56,58
5,0.4,16,55,55,0,TRUE,FALSE,6,52,57
6,0.4,16,55,55,0,TRUE,FALSE,7,51,57
7,0.4,16,55,55,0,TRUE,FALSE,4,53,56
8,0.4,16,55,55,0,TRUE,FALSE,5,55,59
9,0.4,16,55,55,0,FALSE,FALSE,1,56,56
10,0.4,16,55,55,0,TRUE,FALSE,9,53,61
11,0.4,16,55,55,0,TRUE,FALSE,6,53,58
12,0.4,16,55,55,0,TRUE,FALSE,7,52,58
13,0.4,16,55,55,0,TRUE,FALSE,7,51,57
14,0.4,16,55,55,0,TRUE,FALSE,4,52,55
15,0.4,16,55,55,0,TRUE,FALSE,7,51,57
16,0.4,16,55,55,0,TRUE,FALSE,8,52,59
17,0.4,16,55,55,0,FALSE,FALSE,1,56,56
18,0.4,16,55,55,0,TRUE,FALSE,4,53,56
19,0.4,16,55,55,0,TRUE,FALSE,7,52,58
20,0.4,16,55,55,0,FALSE,FALSE,1,57,57
21,0.4,16,55,55,0,TRUE,FALSE,5,53,57
22,0.4,16,55,55,0,FALSE,FALSE,2,56,58
23,0.4,16,55,55,0,TRUE,FALSE,4,53,56
24,0.4,16,55,55,0,TRUE,FALSE,5,53,58
25,0.4,16,55,55,0,TRUE,FALSE,6,54,59
26,0.4,16,55,55,0,TRUE,FALSE,5,54,58
27,0.4,16,55,55,0,TRUE,FALSE,7,52,58
28,0.4,16,55,55,0,TRUE,FALSE,9,51,59
29,0.4,16,55,55,0,TRUE,FALSE,6,52,57
30,0.4,16,55,55,0,TRUE,FALSE,7,52,58
31,0.4,16,55,55,0,TRUE,FALSE,8,53,60
32,0.4,16,55,55,0,TRUE,FALSE,8,53,60
33,0.4,16,55,55,0,TRUE,FALSE,4,54,58
34,0.4,16,55,55,0,TRUE,FALSE,5,53,57
35,0.4,16,55,55,0,TRUE,FALSE,8,52,59
36,0.4,16,55,55,0,TRUE,FALSE,7,53,59
37,0.4,16,55,55,0,TRUE,FALSE,8,52,59
38,0.4,16,55,55,0,TRUE,FALSE,8,52,59
39,0.4,16,55,55,0,TRUE,FALSE,5,54,58
40,0.4,16,55,55,0,TRUE,FALSE,5,54,58
41,0.4,16,55,55,0,TRUE,FALSE,5,54,58
42,0.4,16,55,55,0,TRUE,FALSE,6,54,59
43,0.4,16,55,55,0,TRUE,FALSE,6,52,57
44,0.4,16,55,55,0,TRUE,FALSE,5,53,57
45,0.4,16,55,55,0,TRUE,FALSE,8,52,59
46,0.4,16,55,55,0,TRUE,FALSE,7,53,59
47,0.4,16,55,55,0,TRUE,FALSE,8,52,59
48,0.4,16,55,55,0,TRUE,FALSE,5,52,57
49,0.4,16,55,55,0,TRUE,FALSE,4,53,56
50,0.4,16,55,55,0,TRUE,FALSE,7,51,57
1,0.55,16,55,55,0,TRUE,FALSE,6,52,57
2,0.55,16,55,55,0,TRUE,FALSE,3,54,56
3,0.55,16,55,55,0,TRUE,FALSE,5,53,57
4,0.55,16,55,55,0,TRUE,FALSE,5,53,57
5,0.55,16,55,55,0,FALSE,TRUE,0,NA,NA
6,0.55,16,55,55,0,TRUE,FALSE,6,53,58
7,0.55,16,55,55,0,TRUE,FALSE,7,52,58
8,0.55,16,55,55,0,TRUE,FALSE,6,52,57
9,0.55,16,55,55,0,TRUE,FALSE,4,53,56
10,0.55,16,55,55,0,TRUE,FALSE,4,54,57
11,0.55,16,55,55,0,TRUE,FALSE,3,54,56
12,0.55,16,55,55,0,FALSE,TRUE,0,NA,NA
13,0.55,16,55,55,0,TRUE,FALSE,5,53,57
14,0.55,16,55,55,0,TRUE,FALSE,5,53,57
15,0.55,16,55,55,0,TRUE,FALSE,4,53,56
16,0.55,16,55,55,0,TRUE,FALSE,3,55,57
17,0.55,16,55,55,0,TRUE,FALSE,5,53,57
18,0.55,16,55,55,0,TRUE,FALSE,4,54,57
19,0.55,16,55,55,0,TRUE,FALSE,5,54,58
20,0.55,16,55,55,0,TRUE,FALSE,5,54,58
21,0.55,16,55,55,0,TRUE,FALSE,5,53,57
22,0.55,16,55,55,0,TRUE,FALSE,5,53,57
23,0.55,16,55,55,0,TRUE,FALSE,3,52,55
24,0.55,16,55,55,0,TRUE,FALSE,5,53,57
25,0.55,16,55,55,0,TRUE,FALSE,5,54,58
26,0.55,16,55,55,0,FALSE,FALSE,2,56,57
27,0.55,16,55,55,0,TRUE,FALSE,5,53,58
28,0.55,16,55,55,0,TRUE,FALSE,4,53,56
29,0.55,16,55,55,0,TRUE,FALSE,6,53,58
30,0.55,16,55,55,0,TRUE,FALSE,4,53,56
31,0.55,16,55,55,0,TRUE,FALSE,4,54,57
32,0.55,16,55,55,0,TRUE,FALSE,5,54,58
33,0.55,16,55,55,0,TRUE,FALSE,4,54,57
34,0.55,16,55,55,0,TRUE,FALSE,4,53,56
35,0.55,16,55,55,0,TRUE,FALSE,5,52,56
36,0.55,16,55,55,0,TRUE,FALSE,5,53,57
37,0.55,16,55,55,0,TRUE,FALSE,4,54,57
38,0.55,16,55,55,0,FALSE,FALSE,2,53,54
39,0.55,16,55,55,0,TRUE,FALSE,6,52,57
40,0.55,16,55,55,0,TRUE,FALSE,5,54,58
41,0.55,16,55,55,0,TRUE,FALSE,5,53,57
42,0.55,16,55,55,0,TRUE,FALSE,6,53,58
43,0.55,16,55,55,0,FALSE,FALSE,5,52,57
44,0.55,16,55,55,0,TRUE,FALSE,6,53,58
45,0.55,16,55,55,0,FALSE,TRUE,0,NA,NA
46,0.55,16,55,55,0,TRUE,FALSE,4,54,57
47,0.55,16,55,55,0,TRUE,FALSE,5,54,58
48,0.55,16,55,55,0,TRUE,FALSE,4,54,57
49,0.55,16,55,55,0,TRUE,FALSE,2,54,55
50,0.55,16,55,55,0,TRUE,FALSE,5,53,57
1,0.7,16,55,55,0,TRUE,FALSE,3,54,56
2,0.7,16,55,55,0,TRUE,FALSE,3,54,56
3,0.7,16,55,55,0,TRUE,FALSE,4,54,57
4,0.7,16,55,55,0,TRUE,FALSE,4,54,57
5,0.7,16,55,55,0,TRUE,FALSE,4,54,57
6,0.7,16,55,55,0,TRUE,FALSE,3,54,56
7,0.7,16,55,55,0,TRUE,FALSE,2,54,55
8,0.7,16,55,55,0,TRUE,FALSE,2,54,55
9,0.7,16,55,55,0,TRUE,FALSE,4,54,57
10,0.7,16,55,55,0,FALSE,FALSE,2,56,57
11,0.7,16,55,55,0,TRUE,FALSE,2,54,55
12,0.7,16,55,55,0,TRUE,FALSE,5,54,58
13,0.7,16,55,55,0,TRUE,FALSE,4,54,57
14,0.7,16,55,55,0,TRUE,FALSE,4,54,57
15,0.7,16,55,55,0,TRUE,FALSE,4,54,57
16,0.7,16,55,55,0,TRUE,FALSE,2,54,55
17,0.7,16,55,55,0,TRUE,FALSE,3,54,56
18,0.7,16,55,55,0,TRUE,FALSE,3,55,57
19,0.7,16,55,55,0,TRUE,FALSE,3,54,56
20,0.7,16,55,55,0,TRUE,FALSE,4,54,57
21,0.7,16,55,55,0,TRUE,FALSE,3,53,55
22,0.7,16,55,55,0,TRUE,FALSE,5,53,57
23,0.7,16,55,55,0,TRUE,FALSE,5,53,57
24,0.7,16,55,55,0,TRUE,FALSE,5,53,57
25,0.7,16,55,55,0,FALSE,FALSE,3,54,57
26,0.7,16,55,55,0,TRUE,FALSE,3,55,57
27,0.7,16,55,55,0,FALSE,FALSE,1,54,54
28,0.7,16,55,55,0,TRUE,FALSE,4,54,57
29,0.7,16,55,55,0,TRUE,FALSE,5,53,57
30,0.7,16,55,55,0,TRUE,FALSE,4,53,56
31,0.7,16,55,55,0,TRUE,FALSE,3,54,56
32,0.7,16,55,55,0,TRUE,FALSE,6,53,58
33,0.7,16,55,55,0,TRUE,FALSE,4,54,57
34,0.7,16,55,55,0,TRUE,FALSE,5,53,57
35,0.7,16,55,55,0,TRUE,FALSE,4,54,57
36,0.7,16,55,55,0,TRUE,FALSE,4,54,57
37,0.7,16,55,55,0,TRUE,FALSE,4,54,57
38,0.7,16,55,55,0,TRUE,FALSE,3,54,56
39,0.7,16,55,55,0,TRUE,FALSE,4,54,57
40,0.7,16,55,55,0,TRUE,FALSE,4,54,57
41,0.7,16,55,55,0,TRUE,FALSE,4,54,57
42,0.7,16,55,55,0,TRUE,FALSE,5,53,57
43,0.7,16,55,55,0,TRUE,FALSE,4,54,57
44,0.7,16,55,55,0,TRUE,FALSE,4,54,57
45,0.7,16,55,55,0,TRUE,FALSE,6,53,58
46,0.7,16,55,55,0,FALSE,FALSE,2,54,56
47,0.7,16,55,55,0,TRUE,FALSE,4,54,57
48,0.7,16,55,55,0,TRUE,FALSE,4,54,57
49,0.7,16,55,55,0,TRUE,FALSE,3,54,56
50,0.7,16,55,55,0,TRUE,FALSE,3,54,56
1,0.25,32,55,55,0,TRUE,FALSE,10,50,59
2,0.25,32,55,55,0,TRUE,FALSE,7,51,57
3,0.25,32,55,55,0,TRUE,FALSE,8,52,59
4,0.25,32,55,55,0,TRUE,FALSE,7,52,58
5,0.25,32,55,55,0,TRUE,FALSE,6,52,57
6,0.25,32,55,55,0,FALSE,FALSE,2,52,56
7,0.25,32,55,55,0,TRUE,FALSE,9,51,59
8,0.25,32,55,55,0,TRUE,FALSE,10,52,61
9,0.25,32,55,55,0,TRUE,FALSE,7,52,58
10,0.25,32,55,55,0,TRUE,FALSE,9,52,60
11,0.25,32,55,55,0,TRUE,FALSE,12,51,62
12,0.25,32,55,55,0,TRUE,FALSE,7,53,59
13,0.25,32,55,55,0,TRUE,FALSE,8,53,60
14,0.25,32,55,55,0,TRUE,FALSE,9,51,59
15,0.25,32,55,55,0,TRUE,FALSE,8,52,59
16,0.25,32,55,55,0,TRUE,FALSE,11,50,60
17,0.25,32,55,55,0,TRUE,FALSE,7,53,59
18,0.25,32,55,55,0,TRUE,FALSE,8,51,58
19,0.25,32,55,55,0,TRUE,FALSE,6,53,58
20,0.25,32,55,55,0,TRUE,FALSE,8,52,59
21,0.25,32,55,55,0,TRUE,FALSE,7,52,58
22,0.25,32,55,55,0,TRUE,FALSE,6,51,56
23,0.25,32,55,55,0,TRUE,FALSE,9,50,58
24,0.25,32,55,55,0,TRUE,FALSE,4,55,58
25,0.25,32,55,55,0,TRUE,FALSE,7,52,58
26,0.25,32,55,55,0,TRUE,FALSE,7,52,58
27,0.25,32,55,55,0,TRUE,FALSE,8,50,57
28,0.25,32,55,55,0,TRUE,FALSE,10,50,59
29,0.25,32,55,55,0,TRUE,FALSE,7,51,57
30,0.25,32,55,55,0,FALSE,FALSE,6,51,57
31,0.25,32,55,55,0,TRUE,FALSE,7,52,58
32,0.25,32,55,55,0,FALSE,FALSE,1,57,57
33,0.25,32,55,55,0,FALSE,FALSE,3,52,54
34,0.25,32,55,55,0,TRUE,FALSE,6,51,56
35,0.25,32,55,55,0,TRUE,FALSE,10,51,60
36,0.25,32,55,55,0,TRUE,FALSE,8,52,59
37,0.25,32,55,55,0,FALSE,FALSE,6,52,58
38,0.25,32,55,55,0,TRUE,FALSE,7,51,57
39,0.25,32,55,55,0,TRUE,FALSE,9,52,60
40,0.25,32,55,55,0,TRUE,FALSE,5,53,57
41,0.25,32,55,55,0,TRUE,FALSE,12,50,61
42,0.25,32,55,55,0,TRUE,FALSE,7,53,59
43,0.25,32,55,55,0,TRUE,FALSE,5,52,56
44,0.25,32,55,55,0,TRUE,FALSE,11,50,60
45,0.25,32,55,55,0,TRUE,FALSE,4,52,55
46,0.25,32,55,55,0,FALSE,FALSE,3,57,59
47,0.25,32,55,55,0,TRUE,FALSE,6,53,58
48,0.25,32,55,55,0,TRUE,FALSE,8,54,61
49,0.25,32,55,55,0,TRUE,FALSE,6,53,58
50,0.25,32,55,55,0,FALSE,FALSE,2,52,53
1,0.4,32,55,55,0,TRUE,FALSE,6,53,58
2,0.4,32,55,55,0,TRUE,FALSE,6,53,58
3,0.4,32,55,55,0,TRUE,FALSE,4,54,57
4,0.4,32,55,55,0,TRUE,FALSE,8,52,59
5,0.4,32,55,55,0,TRUE,FALSE,4,53,56
6,0.4,32,55,55,0,TRUE,FALSE,6,53,58
7,0.4,32,55,55,0,TRUE,FALSE,5,53,57
8,0.4,32,55,55,0,TRUE,FALSE,5,53,57
9,0.4,32,55,55,0,TRUE,FALSE,3,55,57
10,0.4,32,55,55,0,TRUE,FALSE,6,53,58
11,0.4,32,55,55,0,TRUE,FALSE,7,52,58
12,0.4,32,55,55,0,TRUE,FALSE,3,54,56
13,0.4,32,55,55,0,TRUE,FALSE,5,54,58
14,0.4,32,55,55,0,TRUE,FALSE,8,52,59
15,0.4,32,55,55,0,FALSE,TRUE,0,NA,NA
16,0.4,32,55,55,0,TRUE,FALSE,5,53,57
17,0.4,32,55,55,0,TRUE,FALSE,5,54,58
18,0.4,32,55,55,0,TRUE,FALSE,6,53,58
19,0.4,32,55,55,0,TRUE,FALSE,6,53,58
20,0.4,32,55,55,0,TRUE,FALSE,6,53,58
21,0.4,32,55,55,0,TRUE,FALSE,6,53,58
22,0.4,32,55,55,0,FALSE,FALSE,1,56,56
23,0.4,32,55,55,0,TRUE,FALSE,6,52,57
24,0.4,32,55,55,0,TRUE,FALSE,4,53,56
25,0.4,32,55,55,0,TRUE,FALSE,5,52,56
26,0.4,32,55,55,0,TRUE,FALSE,5,53,57
27,0.4,32,55,55,0,TRUE,FALSE,4,54,57
28,0.4,32,55,55,0,TRUE,FALSE,5,53,57
29,0.4,32,55,55,0,TRUE,FALSE,5,54,58
30,0.4,32,55,55,0,TRUE,FALSE,6,53,58
31,0.4,32,55,55,0,TRUE,FALSE,5,53,57
32,0.4,32,55,55,0,TRUE,FALSE,5,53,57
33,0.4,32,55,55,0,TRUE,FALSE,5,53,57
34,0.4,32,55,55,0,TRUE,FALSE,4,53,56
35,0.4,32,55,55,0,TRUE,FALSE,5,53,57
36,0.4,32,55,55,0,TRUE,FALSE,6,53,58
37,0.4,32,55,55,0,FALSE,FALSE,1,54,54
38,0.4,32,55,55,0,TRUE,FALSE,7,52,58
39,0.4,32,55,55,0,TRUE,FALSE,5,53,57
40,0.4,32,55,55,0,TRUE,FALSE,7,52,58
41,0.4,32,55,55,0,TRUE,FALSE,6,53,58
42,0.4,32,55,55,0,TRUE,FALSE,5,52,56
43,0.4,32,55,55,0,TRUE,FALSE,5,53,57
44,0.4,32,55,55,0,TRUE,FALSE,6,52,57
45,0.4,32,55,55,0,TRUE,FALSE,7,53,59
46,0.4,32,55,55,0,TRUE,FALSE,5,53,57
47,0.4,32,55,55,0,TRUE,FALSE,6,52,57
48,0.4,32,55,55,0,TRUE,FALSE,5,53,57
49,0.4,32,55,55,0,TRUE,FALSE,4,53,56
50,0.4,32,55,55,0,TRUE,FALSE,6,54,59
1,0.55,32,55,55,0,FALSE,FALSE,2,54,56
2,0.55,32,55,55,0,TRUE,FALSE,4,54,57
3,0.55,32,55,55,0,TRUE,FALSE,5,53,57
4,0.55,32,55,55,0,TRUE,FALSE,3,54,56
5,0.55,32,55,55,0,TRUE,FALSE,4,53,56
6,0.55,32,55,55,0,FALSE,FALSE,4,53,57
7,0.55,32,55,55,0,TRUE,FALSE,5,53,57
8,0.55,32,55,55,0,TRUE,FALSE,4,54,57
9,0.55,32,55,55,0,TRUE,FALSE,2,55,56
10,0.55,32,55,55,0,TRUE,FALSE,4,53,56
11,0.55,32,55,55,0,TRUE,FALSE,4,53,56
12,0.55,32,55,55,0,TRUE,FALSE,4,54,57
13,0.55,32,55,55,0,TRUE,FALSE,4,53,56
14,0.55,32,55,55,0,TRUE,FALSE,4,54,57
15,0.55,32,55,55,0,TRUE,FALSE,3,54,56
16,0.55,32,55,55,0,TRUE,FALSE,4,54,57
17,0.55,32,55,55,0,TRUE,FALSE,5,53,57
18,0.55,32,55,55,0,FALSE,FALSE,1,54,54
19,0.55,32,55,55,0,TRUE,FALSE,5,53,57
20,0.55,32,55,55,0,TRUE,FALSE,3,55,57
21,0.55,32,55,55,0,TRUE,FALSE,4,54,57
22,0.55,32,55,55,0,TRUE,FALSE,4,54,57
23,0.55,32,55,55,0,TRUE,FALSE,3,54,56
24,0.55,32,55,55,0,TRUE,FALSE,5,53,57
25,0.55,32,55,55,0,TRUE,FALSE,2,54,55
26,0.55,32,55,55,0,TRUE,FALSE,4,54,57
27,0.55,32,55,55,0,TRUE,FALSE,4,54,57
28,0.55,32,55,55,0,FALSE,FALSE,1,54,54
29,0.55,32,55,55,0,TRUE,FALSE,4,54,57
30,0.55,32,55,55,0,TRUE,FALSE,3,55,57
31,0.55,32,55,55,0,TRUE,FALSE,4,54,57
32,0.55,32,55,55,0,TRUE,FALSE,5,53,57
33,0.55,32,55,55,0,TRUE,FALSE,5,53,57
34,0.55,32,55,55,0,FALSE,FALSE,2,54,56
35,0.55,32,55,55,0,TRUE,FALSE,5,53,57
36,0.55,32,55,55,0,TRUE,FALSE,6,53,58
37,0.55,32,55,55,0,TRUE,FALSE,3,54,56
38,0.55,32,55,55,0,TRUE,FALSE,3,54,56
39,0.55,32,55,55,0,TRUE,FALSE,3,53,55
40,0.55,32,55,55,0,TRUE,FALSE,6,53,58
41,0.55,32,55,55,0,TRUE,FALSE,5,53,57
42,0.55,32,55,55,0,TRUE,FALSE,5,53,57
43,0.55,32,55,55,0,TRUE,FALSE,3,54,56
44,0.55,32,55,55,0,TRUE,FALSE,3,54,56
45,0.55,32,55,55,0,TRUE,FALSE,4,54,57
46,0.55,32,55,55,0,TRUE,FALSE,4,53,56
47,0.55,32,55,55,0,TRUE,FALSE,5,53,57
48,0.55,32,55,55,0,FALSE,TRUE,0,NA,NA
49,0.55,32,55,55,0,TRUE,FALSE,4,53,56
50,0.55,32,55,55,0,TRUE,FALSE,4,54,57
1,0.7,32,55,55,0,TRUE,FALSE,3,54,56
2,0.7,32,55,55,0,TRUE,FALSE,4,54,57
3,0.7,32,55,55,0,TRUE,FALSE,4,53,56
4,0.7,32,55,55,0,TRUE,FALSE,3,54,56
5,0.7,32,55,55,0,TRUE,FALSE,2,54,55
6,0.7,32,55,55,0,TRUE,FALSE,5,53,57
7,0.7,32,55,55,0,TRUE,FALSE,4,54,57
8,0.7,32,55,55,0,TRUE,FALSE,2,54,55
9,0.7,32,55,55,0,TRUE,FALSE,4,54,57
10,0.7,32,55,55,0,TRUE,FALSE,4,54,57
11,0.7,32,55,55,0,FALSE,FALSE,1,56,56
12,0.7,32,55,55,0,TRUE,FALSE,3,54,56
13,0.7,32,55,55,0,TRUE,FALSE,4,54,57
14,0.7,32,55,55,0,TRUE,FALSE,3,54,56
15,0.7,32,55,55,0,TRUE,FALSE,3,54,56
16,0.7,32,55,55,0,TRUE,FALSE,3,54,56
17,0.7,32,55,55,0,FALSE,FALSE,2,56,57
18,0.7,32,55,55,0,TRUE,FALSE,3,54,56
19,0.7,32,55,55,0,TRUE,FALSE,5,53,57
20,0.7,32,55,55,0,TRUE,FALSE,4,54,57
21,0.7,32,55,55,0,TRUE,FALSE,5,53,57
22,0.7,32,55,55,0,TRUE,FALSE,2,55,56
23,0.7,32,55,55,0,TRUE,FALSE,5,53,57
24,0.7,32,55,55,0,TRUE,FALSE,3,54,56
25,0.7,32,55,55,0,TRUE,FALSE,3,54,56
26,0.7,32,55,55,0,TRUE,FALSE,4,54,57
27,0.7,32,55,55,0,TRUE,FALSE,3,54,56
28,0.7,32,55,55,0,FALSE,FALSE,1,54,54
29,0.7,32,55,55,0,TRUE,FALSE,3,54,56
30,0.7,32,55,55,0,TRUE,FALSE,3,54,56
31,0.7,32,55,55,0,TRUE,FALSE,4,54,57
32,0.7,32,55,55,0,TRUE,FALSE,3,54,56
33,0.7,32,55,55,0,FALSE,TRUE,0,NA,NA
34,0.7,32,55,55,0,TRUE,FALSE,4,54,57
35,0.7,32,55,55,0,TRUE,FALSE,4,54,57
36,0.7,32,55,55,0,TRUE,FALSE,3,54,56
37,0.7,32,55,55,0,TRUE,FALSE,3,54,56
38,0.7,32,55,55,0,TRUE,FALSE,3,54,56
39,0.7,32,55,55,0,TRUE,FALSE,4,54,57
40,0.7,32,55,55,0,TRUE,FALSE,4,54,57
41,0.7,32,55,55,0,TRUE,FALSE,4,54,57
42,0.7,32,55,55,0,TRUE,FALSE,4,54,57
43,0.7,32,55,55,0,TRUE,FALSE,4,54,57
44,0.7,32,55,55,0,TRUE,FALSE,3,54,56
45,0.7,32,55,55,0,TRUE,FALSE,3,54,56
46,0.7,32,55,55,0,TRUE,FALSE,3,54,56
47,0.7,32,55,55,0,TRUE,FALSE,3,54,56
48,0.7,32,55,55,0,TRUE,FALSE,3,54,56
49,0.7,32,55,55,0,TRUE,FALSE,4,53,56
50,0.7,32,55,55,0,TRUE,FALSE,3,54,56
\end{filecontents*}
\csvreader[
    longtable=rrrrrrrrrrr,
    table head=
        \caption{\textsc{Monte Carlo Results}}\label{tab:monte_carlo_results} \\
        \toprule
        rep\_id & effect & m & tau0 & tau\_hat & abs\_error & covered & cs\_empty & cs\_size & cs\_min & cs\_max \\
        \midrule\endfirsthead
        \caption*{\textsc{Monte Carlo Results} (continued)} \\
        \toprule
        rep\_id & effect & m & tau0 & tau\_hat & abs\_error & covered & cs\_empty & cs\_size & cs\_min & cs\_max \\
        \midrule\endhead
        \bottomrule\endfoot,
    late after line=\\
]{monte_carlo_results.csv}{1=\colA, 2=\colB, 3=\colC, 4=\colD, 5=\colE, 6=\colF, 7=\colG, 8=\colH, 9=\colI, 10=\colJ, 11=\colK}%
{\colA & \colB & \colC & \colD & \colE & \colF & \colG & \colH & \colI & \colJ & \colK}
\endgroup
\end{document}